\documentclass[journal,12pt,onecolumn,draftclsnofoot]{IEEEtran}

\usepackage{cite}
\usepackage{graphicx}
\usepackage{algorithm,algorithmic}
\usepackage[caption=false, font=footnotesize]{subfig}
\usepackage{url}
\usepackage{mathrsfs}
\usepackage{amsfonts}
\usepackage{amsmath}

\DeclareMathOperator*{\argmin}{arg\,min}
\usepackage{amssymb}
\usepackage{mathtools}
\usepackage{algorithm,algorithmic}
\usepackage{amsthm}
\newtheorem{theorem}{Theorem}

\usepackage{bm}
\renewcommand{\vec}[1]{\boldsymbol{#1}}

\usepackage{etoolbox}
\usepackage{xcolor}
\usepackage{multirow}
\usepackage{dsfont}

\usepackage{etoolbox}
\AtBeginEnvironment{algorithm}{\linespread{1.1}\selectfont}

\begin{document}

\title{Multi-user 
Co-inference with Batch Processing Capable Edge Server}

\author{Wenqi Shi,
        Sheng Zhou,~\IEEEmembership{Member,~IEEE}, and
        Zhisheng Niu,~\IEEEmembership{Fellow,~IEEE}, \linebreak Miao Jiang, and Lu Geng

\thanks{This work is sponsored in part by the National Key R\&D Program of China 2018YFB1800800, the Nature Science Foundation of China (No. 61871254, No. 61861136003), and Hitachi Ltd.}
\thanks{W. Shi, S. Zhou, and Z. Niu are with the Beijing National
Research Center for Information Science and Technology, Department of Electronic Engineering, Tsinghua University, Beijing 100084,
China (e-mail: swq17@mails.tsinghua.edu.cn; sheng.zhou@tsinghua.edu.cn;
niuzhs@tsinghua.edu.cn).}
\thanks{M. Jiang and L. Geng are with Hitachi (China) Research \& Development Cooperation, Beijing 100190, China (e-mail: miaojiang@hitachi.cn; lgeng@hitachi.cn).}
}

\maketitle

\begin{abstract}
Graphics processing units (GPUs) can improve deep neural network inference throughput via batch processing, where multiple tasks are concurrently processed. 
\textcolor{black}{
We focus on novel scenarios that the energy-constrained mobile devices offload inference tasks to an edge server with GPU.
The inference task is partitioned into sub-tasks for a finer granularity of offloading and scheduling, and the user energy consumption minimization problem under inference latency constraints is investigated.}
\textcolor{black}{To deal with the coupled offloading and scheduling introduced by concurrent batch processing, we first consider an offline problem with a constant edge inference latency and the same latency constraint.
It is proven that optimizing the offloading policy of each user independently and aggregating all the same sub-tasks in one batch is optimal, and thus the independent partitioning and same sub-task aggregating (IP-SSA) algorithm is inspired.
Further, the optimal grouping (OG) algorithm is proposed to optimally group tasks when the latency constraints are different.}
Finally, when future task arrivals cannot be precisely predicted, a deep deterministic policy gradient (DDPG) agent is trained to call OG.
\textcolor{black}{Experiments show that IP-SSA reduces up to 94.9\% user energy consumption in the offline setting, while DDPG-OG outperforms DDPG-IP-SSA by up to 8.92\% in the online setting.}

\end{abstract}


\section{Introduction}
With recent breakthroughs in deep learning, deep neural networks (DNNs) have been successfully applied in a wide range of artificial intelligence (AI) services and applications, including natural language processing \cite{otter2020survey}, autonomous driving \cite{grigorescu2020survey}, and content recommendation \cite{da2020recommendation}.
In the meantime, driven by the rapid development of sensing and communication capabilities of mobile devices (e.g., smartphones and autonomous driving vehicles), a large amount of data are generated.
Such AI applications usually have stringent latency constraints, which makes it impractical to
upload the distributed raw data to cloud servers for centralized processing, due to the limited network bandwidth.
Therefore, a new yet promising research area, called edge AI or edge learning has been unleashed, putting DNNs to the network edge \cite{sun2020edge}\cite{gunduz2020communicate}.

\textcolor{black}{Recently, graphics processing units (GPUs) and custom-designed DNN hardware accelerators (e.g., neural-network processing units, NPUs) have been widely used to accelerate the computing of DNNs, due to their great parallel computing capabilities.}
\emph{Batch processing} is an essential technique to better utilize parallelism,
where multiple tasks are aggregated into a batch and concurrently processed.
As a result, the throughput of DNN computing can be improved, for both inference \cite{choi2021lazy} and training \cite{you2019large}.

\textcolor{black}{
Conducting timely and reliable DNN inference at the resource-constrained network edge is challenging due to the high computation energy consumption.
For example, for L5 level autonomous driving vehicles, the connected and autonomous driving subsystem will consume 1200 Watts, and the major part is contributed by the computation \cite{brost2021energy}.
As a result, the autonomous driving can reduce electric vehicle range by up to 15\% \cite{mohan2020trade}.
To address this issue, we consider the scenarios that mobile devices can partition the DNN inference task into sub-tasks and offload a part of the sub-tasks to the edge server to save energy.
Motivated by fact that in applications like autonomous driving and smart manufacturing, mobile devices run the same DNN inference task (e.g., object detection DNNs), the edge server can be equipped with GPUs to improve the throughput of the offloaded sub-tasks via batch processing.
In such scenarios, the inference task offloading and scheduling policy is critical to reduce the user energy consumption.
}

\textcolor{black}{
In the literature, vast research efforts have been made to reduce the energy consumption of mobile users and ensure latency constraints of edge inference.}
Some researchers focus on reducing the on-board inference latency by DNN model compression \cite{deng2020model} or DNN hardware accelerators \cite{sze2017efficient}.
However, we follow an orthogonal vein of research, which is known as device-edge co-inference.
Neurosurgeon \cite{kang2017neurosurgeon} proposes the first co-inference framework, where the DNN is partitioned into two parts with layer-level granularity, and the later part is offloaded to the edge server.
\textcolor{black}{The main challenge is that the intermediate data size of DNNs can be large, making offloading bandwidth-consuming.}
Therefore, the authors of \cite{ko2018edge} propose JPEG-based feature coding to compress the intermediate data. 
Recently, joint source and channel coding (JSCC) for co-inference under noisy wireless channels has attracted research attentions \cite{shao2020communication, jankowski2020joint, shao2021learning}, where the feature compression and channel coding are jointly optimized.
On the other hand, some researchers modify the DNN architecture to avoid transmitting large intermediate data.
Edgent \cite{li2019edge} adds early-exit branches to the original DNN, so that the result with high confidence can be directly output.
Since the size of the result is smaller than that of the intermediate data,
such mechanism can reduce the transmission latency in co-inference, and is further extended to the graph neural networks (GNNs) in \cite{shao2021branchy}. 

However, all aforementioned works focus on the one-shot problem for single-user scenarios, while an edge server may serve multiple users with random task arrivals in realistic scenarios.
In \cite{hu2019dynamic}, the DNN inference tasks arrive at a fixed frequency.
The authors minimize the co-inference latency under light workload scenarios or maximize the co-inference throughput under heavy workload scenarios.
Further, the authors of \cite{song2021adaptive} extend Edgent to the scenarios with random task arrivals.
They propose a deep reinforcement learning agent to balance the trade-off between the number of tasks completed within the latency constraints and the inference accuracy.
Nevertheless, the edge server processes the offloaded tasks in a first-in-first-out (FIFO) manner in \cite{hu2019dynamic, song2021adaptive}.
\textcolor{black}{To the best of our knowledge, only \cite{tang2020joint} considers multi-user co-inference, wherein the DNN partitioning and edge server CPU cores allocation are jointly optimized to minimize the maximum inference latency among users.}
\textcolor{black}{However, for the promising DNN hardware accelerators such as GPUs and NPUs, \cite{hu2019dynamic, song2021adaptive, tang2020joint} cannot be directly applied, since the computing resource allocation and inference task scheduling would be completely different due to batch processing.}

The problem of batching inference tasks has been studied in the context of cloud computing, which is known as the DNN inference serving system \cite{zhang2019mark}.
Since inference tasks arrive at different time, batching policy should carefully balance the trade-off between the latency and throughput.
For instance, large batches can improve the throughput, while suffering from long latency of waiting for enough tasks.
In most existing works, batch processing is triggered according to the following  system parameters, batch size (i.e., the number of arrived tasks) and time window (i.e., the elapsed time after the last batch).
\textcolor{black}{To improve the performance, the thresholds of batch size and time window can be optimized \cite{ali2020batch, crankshaw2017clipper}.}
Nevertheless, since the whole inference task is offloaded to the cloud and the offloading latency is omitted in DNN inference serving systems, these approaches are not suitable for device-edge co-inference.

Motivated by the above observations, the problem of joint inference task offloading and scheduling for co-inference with batch processing capable edge server is important yet open.
\textcolor{black}{In this work, 
our goal is to minimize the energy consumption of mobile users while ensuring the latency constraints of the inference tasks.}
Due to the concurrent nature of batch processing, the main challenge is the coupling of the offloading decisions and the scheduling policy.
Since only the same sub-tasks can be aggregated into a batch, whether a user can offload a typical sub-task to the edge server not only depends on how many sub-tasks other users offload, but also depends on \emph{what sub-tasks other users offload}.
Further, since a batch cannot be started until all sub-tasks in the batch are ready, the starting time of a batch depends on the local computing capabilities and channel conditions of users, which can be time varying.

To solve this problem, we begin with an offline scenario as \cite{tang2020joint}.
\textcolor{black}{A joint offloading and scheduling problem is formulated, which is a mixed-integer non-linear programming problem and is NP-hard in general.}
\textcolor{black}{
To address the coupling between offloading and scheduling, we introduce two mild simplifications:}
1) for a typical sub-task, the inference latency of the edge server is not related to the batch size, and 2) all tasks have the same latency constraint.
For the simplified problem, we theoretically prove that each user independently decides its offloading decision, and the edge server aggregates all the same offloaded sub-tasks in one batch is optimal.
Then, we show that the optimal solution can be extended to two sophisticated scenarios that remove the simplifications, and accordingly propose two low-complexity algorithms (i.e., the IP-SSA algorithm and the OG algorithm).
\textcolor{black}{When the edge inference latency increases with the batch size, IP-SSA searches all potential batch size to derive the best batch size that minimizes the user energy consumption without violating the latency constraint.}
When tasks have different latency constraints, OG derives the optimal grouping policy via dynamic programming, where the offloading and scheduling policy of each group is given by IP-SSA.
Finally, we consider an online scenario that the arrivals of future tasks cannot be precisely predicted.
In order to seek a balance between serving arrived tasks and reserving resources for future tasks, a reinforcement learning (RL) agent is trained to schedule the offline algorithm, and a two-dimensional action space is proposed.
Our main contributions are summarized as follows.
\begin{itemize}
  \item{\textcolor{black}{To the best of our knowledge, this work is the first to consider the multi-user co-inference scenario with batch processing capable edge server. 
  We develop a formal framework to jointly optimize the offloading and scheduling to minimize the user energy consumption under inference latency constraints.}}
  \item{
  \textcolor{black}{For offline scenarios that all tasks have arrived and latency constraints are known, we analyze the structural properties of the optimal solution under two mild simplifications.
  Further, IP-SSA and OG are inspired for realistic scenarios that remove the simplifications.}}
  \item{
  For online scenarios, an RL agent is proposed. A two dimentional action space is designed to balance two key trade-offs: the trade-off between waiting latency and batch size, and the trade-off between processing time and idle period.
  }
  \item{We conduct extensive experiments on two DNNs, mobilenet-v2 \cite{sandler2018mobilenetv2} and 3dssd \cite{yang20203dssd}. 
  \textcolor{black}{The experiment results show that by utilizing batch processing, the proposed methods can greatly reduce the user energy consumption compared to conventional edge server computation resource allocation methods, especially when the number of users is large.}}
\end{itemize}

The remainder of this paper is organized as follows.
In Section II, we introduce the system model and formulate the problem.
\textcolor{black}{The structural properties of the optimal solution under two mild simplifications are provided in Section III.}
In Section IV, low-complexity algorithms for realistic offline scenarios are proposed, and further an RL agent is trained for the online scenario.
The experiment results are shown in Section V and we conclude the paper in Section VI.

\section{System Model and Problem Formulation}

\begin{figure}[!t] 
\setlength{\abovecaptionskip}{2pt}
\setlength{\belowcaptionskip}{2pt}
\centering 
\includegraphics[width=0.5\linewidth]{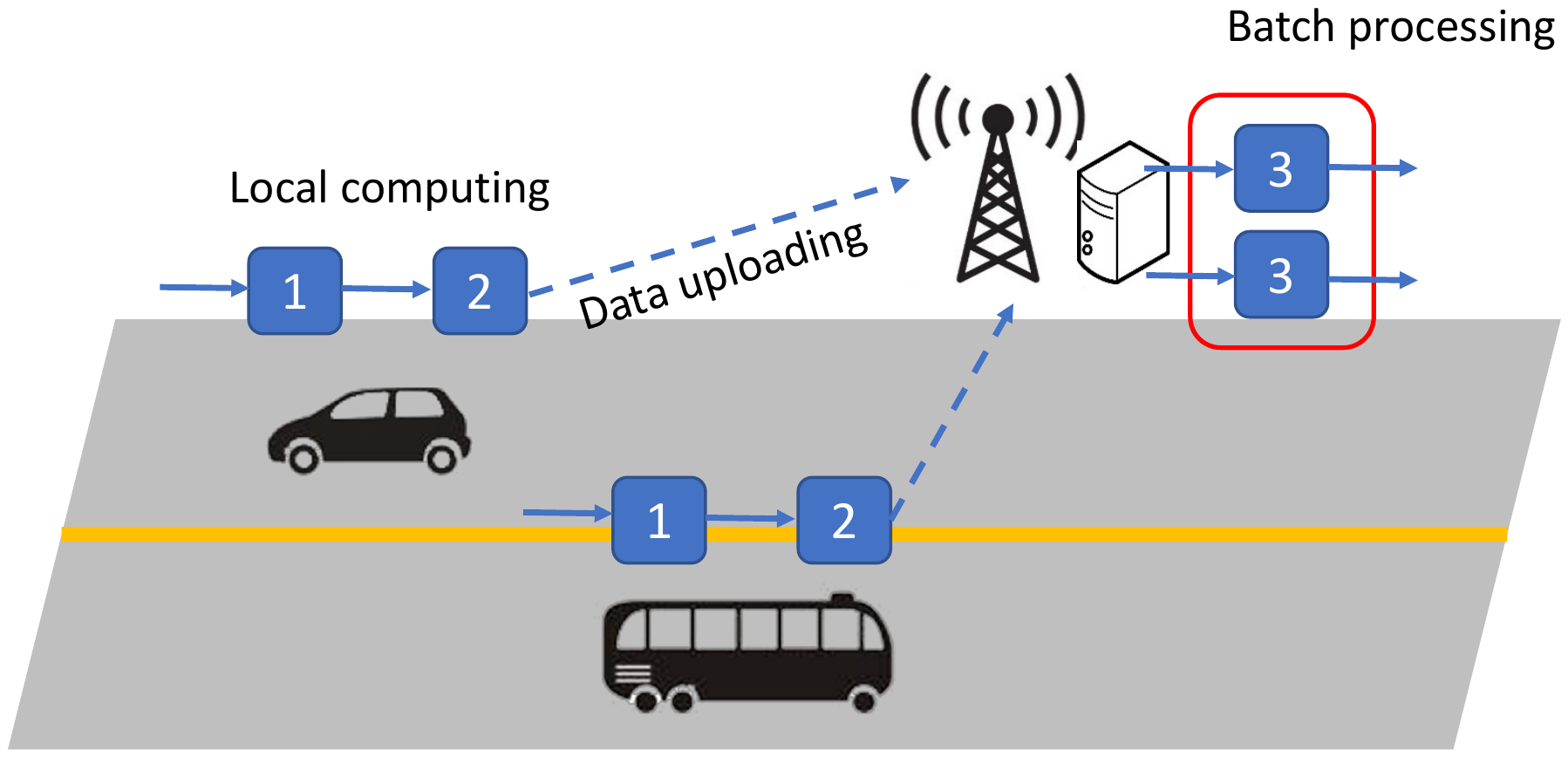} 
\caption{An illustration of the device-edge co-inference.} 
\label{system}
\vspace{-10pt}
\end{figure}

\textcolor{black}{Consider a device-edge co-inference system, consisting of $M$ mobile users and an edge server with one GPU \footnote{\textcolor{black}{For simplicity, we assume that the edge server is equipped with only one GPU. However, by assigning users to different GPUs, the proposed algorithm can be easily extended to the multiple GPUs scenario.}}.}
As shown in Fig. \ref{system}, all users have the same pre-trained DNN model, and perform inference tasks on their individual input data.
We use $l_m$ to denote the latency constraint of user $m$.
\textcolor{black}{Since the mobile devices should be functional without an edge server, we assume that each mobile device contain the entire DNN, and can locally complete the inference task within the latency constraint with high energy consumption.}
To reduce the energy consumption, each user can partition the whole DNN inference task into several sub-tasks, and offload a part of the sub-tasks to nearby edge servers.
For the offloaded sub-tasks, we assume that the edge server is equipped with one GPU, so the same sub-tasks from different users can be processed via batch processing to improve the throughput \cite{hanhirova2018latency, ren2020accelerating}.
The main notations are summarized in Table \ref{notation}.

\begin{table}[!t]
\setlength{\abovecaptionskip}{2pt}
\setlength{\belowcaptionskip}{2pt}
\caption{Summary of Main Notations}
\label{notation}
\begin{center}
\begin{tabular}{c p{0.68\linewidth}}
\hline
\textbf{Notation} & \textbf{Definition}\\
\hline
$M$; $l_m$ & Number of mobile users; Latency constraint \\
$N$; $A_n$; $B_n$ & Number of sub-tasks; Computing workload; Output data size\\
$x_{m,n,k}$; $f_m$ & Offloading and scheduling decision; Local computing frequency\\
$y_{m,n}^\text{u}$; $y_{m,n}^\text{d}$ & Indicator of data uploading; Indicator of data downloading\\
$R_m^\text{u}$; $R_m^\text{d}$ & Upload transmission rate; Download transmission rate\\
$p_m^\text{u}$; $p_m^\text{d}$ & Upload power consumption; Download power consumption\\
$b_{n,k}$; $F_n(\cdot)$ & Batch size; The function of edge computing latency w.r.t. batch size\\
$t_{m,n}$; $s_k$ &  Completing time; Batch starting time\\
\hline
\end{tabular}
\end{center}
\vspace{-10pt}
\end{table}

\subsection{DNN Inference Task Model}
We model the DNN inference task as \emph{a sequence of sub-tasks} in this paper, due to the fact that most modern DNNs are constructed by some basic layers
(e.g., convolution layer, fully-connected layer, batch-normalization layer).
Notice that some DNN modules (e.g., residual block \cite{he2016deep}) are not fully sequential, since they have parallel layers or bypass structures.
However, we can still use the sequential model by abstracting one or several consecutive modules as a sub-task, which also has been adopted in the literature \cite{tang2020joint}.
We assume the DNN inference task has $N$ sequential sub-tasks.
For a typical sub-task $n \in \{1, \dots, N\}$, $A_n$ is used to denote the computation workload of sub-task $n$.
We use $B_n$ to denote the output data size of the $n$-th sub-task, which is also the input data size of the $(n+1)$-th sub-task.
Furthermore, $B_0$ is used to denote the size of the input data of the first sub-task.

In the experiments, we use two different DNNs: a light-weight image classification DNN (mobilenet-v2 \cite{sandler2018mobilenetv2}), and a heavy-weight point cloud object detection DNN (3dssd \cite{yang20203dssd}). 
The partition of mobilenet-v2 and 3dssd is shown in Fig. \ref{arch}. \footnote{In this work, we put a partition point after each module.
Note that the DNNs can be partitioned at a finer granularity, say layers. However, this may increase the complexity of the problem, and is beyond the scope of this paper. 
}

\begin{figure}[!t] 
\setlength{\abovecaptionskip}{2pt}
\setlength{\belowcaptionskip}{2pt}
\centering 
\includegraphics[width=0.45\linewidth, angle=90]{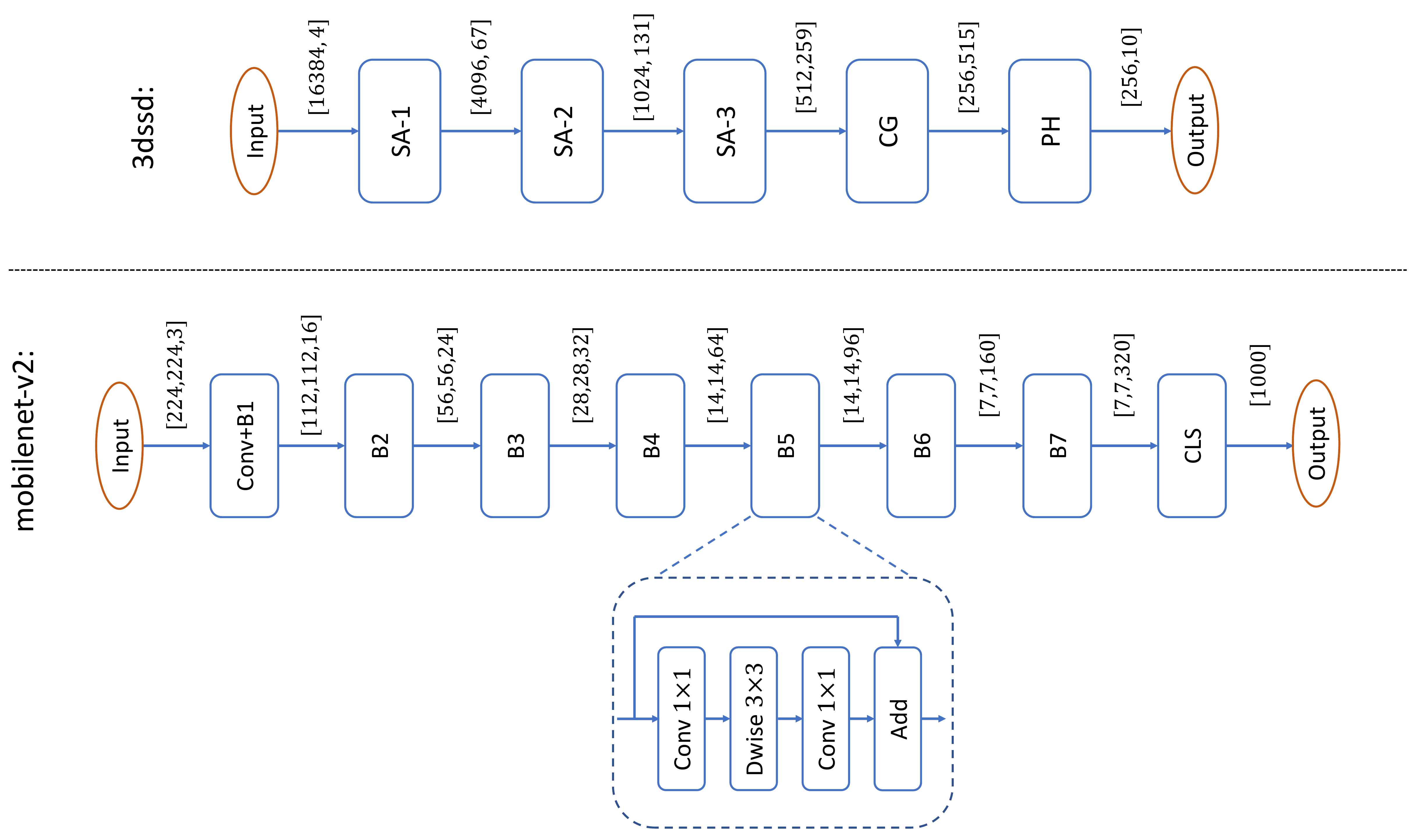} 
\caption{DNN architectures and partition points used in the experiments. For 3dssd, SA, CG, and PH are the abbreviations for set abstraction module, candidate generation module, and prediction head module. For mobilenet-v2, Conv, B, and CLS are the abbreviations for convolution layer, bottleneck module, and classification layer. The architecture of the fifth bottleneck module of mobilenet-v2 is illustrated, while the details of other modules can be found in the original papers \cite{sandler2018mobilenetv2,yang20203dssd}.
Further, Fig. \ref{arch} also shows the shape of the output data of each sub-task, which can be used to calculate the intermediate data size $B_n$.} 
\label{arch}
\vspace{-8pt}
\end{figure}

\subsection{Co-inference Model}
A binary variable $x_{m,n,k} \in \{0, 1\}$ is used to denote where and how each sub-task is processed.
Specifically, $x_{m,n,0} = 1$ indicates that the $n$-th sub-task of the $m$-th  user is processed locally.
Otherwise, the sub-task is offloaded to the edge server, and $x_{m,n,0} = 0$.
The value of $x_{m,n,k}, \forall k \geq 1$ and the formal definition of $x_{m,n,k}$ will be introduced in the next sub-section.

\subsubsection{Local Computation Model}
We use $f_m$ to denote the computing frequency of user $m$, which can be adjusted by applying the Dynamic Voltage and Frequency Scaling (DVFS) technique \cite{rabaey2003digital}. 
Since the computation workload of the $n$-th sub-task is denoted by $A_n$, the local computing latency is given by 
\begin{equation}
    l_{m,n}^\text{cp}(f_m) = \frac{A_n}{f_m}, \  f_{m,\text{min}} \leq f_m \leq f_{m,\text{max}}, \label{tcp}
\end{equation}
and the corresponding computing energy consumption is
\begin{equation}
    e_{m,n}^\text{cp}(f_m) = \kappa_m A_n f_m^2, \label{ecp}
\end{equation}
where $\kappa_m$ is the effective switched capacitance that depends
on the chip architecture of device $m$ \cite{yang2020energy}.
\textcolor{black}{Since the mobile devices should be able to locally complete the inference task within the latency constraint, we have $\frac{\sum \limits_{n=1}^N A_n}{f_{m,\text{max}}} \leq l_m $.}

\subsubsection{Intermediate Data Transmission Model}
Use $R_m^\text{u}$ to denote the uplink transmission rate of user $m$, then the transmission latency of uploading the output data of the $n$-th sub-task is 
\begin{equation}
    l_{m,n}^\text{u} = \frac{B_n}{R_m^\text{u}}. \label{tu}
\end{equation}
The corresponding uploading energy consumption is
\begin{equation}
    e_{m,n}^\text{u} = l_{m,n}^\text{u} p_m^\text{u}, \label{eu}
\end{equation}
where $p_m^\text{u}$ is the power consumption of the transmitter.

Similarly, the latency of downloading the intermediate output data of the $n$-th sub-task and corresponding energy consumption is given by $l_{m,n}^\text{d} = \frac{B_n}{R_m^\text{d}}$, and $e_{m,n}^\text{d} = l_{m,n}^\text{d} p_m^\text{d}$, respectively, where $R_m^\text{d}$ is the downlink transmission rate and $p_m^\text{d}$ is the power consumption of the receiver.

\subsection{Batch Processing Model}
For the offloaded sub-tasks, $x_{m,n,k} = 1, k\geq 1$ indicates that the sub-task is processed in the $k$-th batch by the edge server.
Obviously, each sub-task only needs to be processed once, and thus the maximum number of batches is $MN$.
Therefore, we have $ 1 \leq k \leq MN$.
Then, the offloading and scheduling decision variable $x_{m,n,k}$ can be formally defined as follows
\begin{equation}
x_{m,n,k}=\left\{
\begin{array}{lr}
1, \  \left\{
    \begin{array}{ll}
        \text{$k=0$, if processed locally}, \\
        \text{$1\leq k\leq MN$, if processed in the $k$-th batch by the edge server},
	\end{array}
	\right. \\
0, \  \text{otherwise}.
\end{array}
\right.
\end{equation}
Since each sub-task needs and only needs to be processed once, we have
\begin{equation}
    \sum_{k=0}^{M N} x_{m,n,k}=1. \label{C11}
\end{equation}
We use $b_{n,k}$ to denote the batch size of the $n$-th sub-task in the $k$-th batch, which is given by
\begin{equation}
    \sum_{m=1}^M x_{m,n,k}=b_{n,k}. \label{C12}
\end{equation}
In this work, we assume that only the same sub-task can be processed in one batch, and thus 
\begin{equation}
    b_{n,k} b_{n', k} = 0, \  \forall 1 \leq k \leq MN \; \text{and}\; n\neq n'. \label{C13}
\end{equation}

Further, $t_{m,n}$ is used to denote the completing time of the $n$-th sub-task of the $m$-th user, and $s_k$ is used to denote the starting time of the $k$-th batch.
Since a batch cannot be started until all sub-tasks in that batch are ready to be processed, we have
\begin{equation}
    s_k \geq t_{m,n-1} + \frac{B_{n-1}}{R_m^\text{u}} y_{m,n-1}^\text{u},  \ \forall (m,n,k) \in \{(m,n,k)| x_{m,n,k} = 1 \ \text{and} \ k\geq 1\}, \label{C14}
\end{equation}
where a sub-task are not ready until the previous sub-task is completed and the input data is uploaded (if needed). 
Here, we use a binary variable $y_{m,n-1}^\text{u}$ (defined in \eqref{C15}) to denote whether the 
output intermediate data of the $(n-1)$-th sub-task needs to be uploaded.
\begin{equation}
    y_{m,n-1}^\text{u} = \mathds{1} \{x_{m,n,0} - x_{m,n-1,0} < 0\}. \label{C15}\\
\end{equation}

We use a set of functions $F_n(\cdot)$ to characterize the relation between the computing latency of the edge server and the batch size for each sub-task. Fig. \ref{l_b} shows the profiling results of mobilenet-v2 and 3dssd on a GPU server, where the blue curves show $F_n(\cdot)$ for each sub-task and the red curves show the throughput improvement brought by batch processing.
\begin{figure}[!t]
\setlength{\abovecaptionskip}{2pt}
\setlength{\belowcaptionskip}{2pt}
\centering
\subfloat[]{
\includegraphics[width=0.48\linewidth]{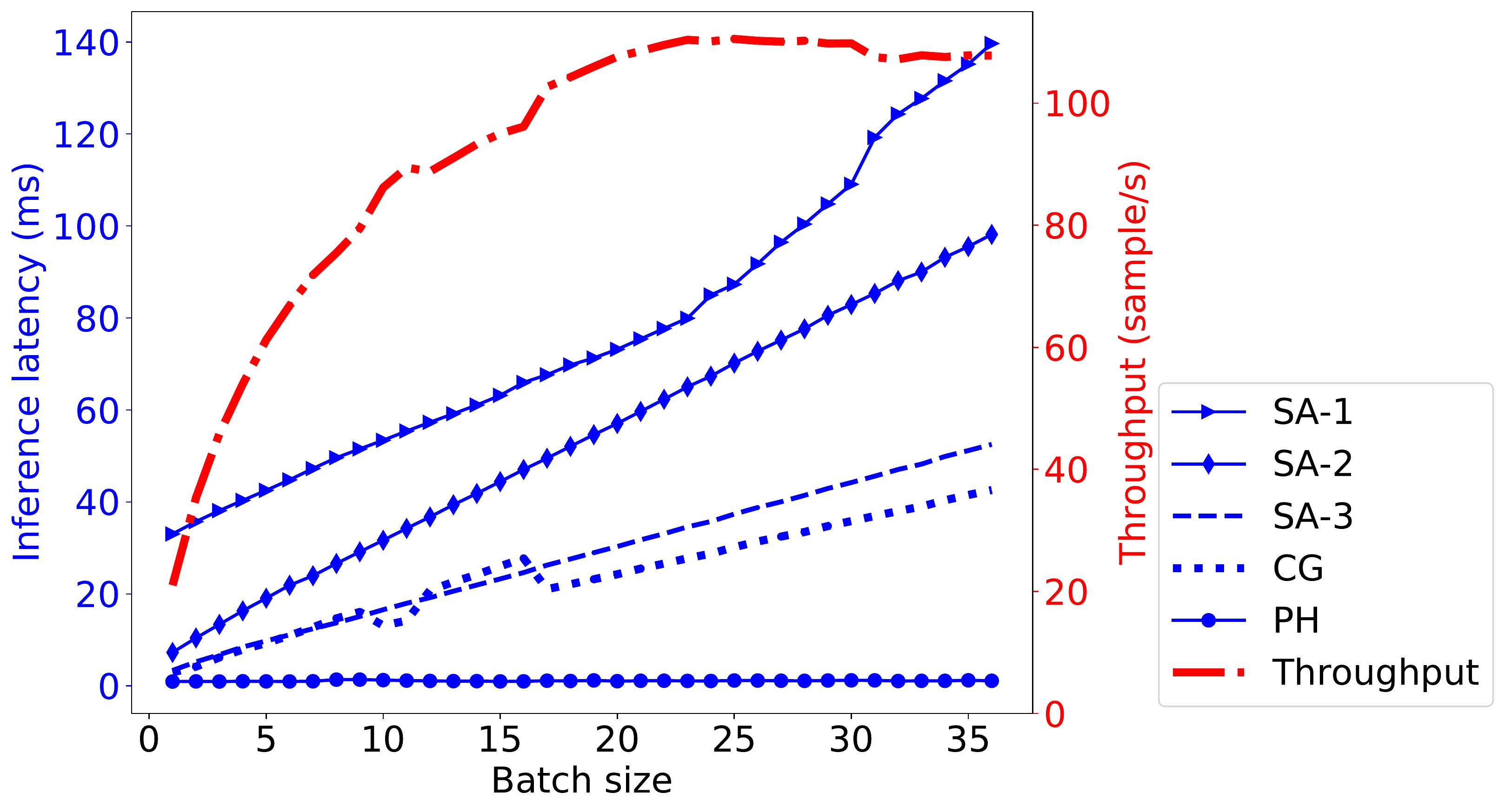}}
\label{l_b_3dssd}\hfill
\subfloat[]{
\includegraphics[width=0.48\linewidth]{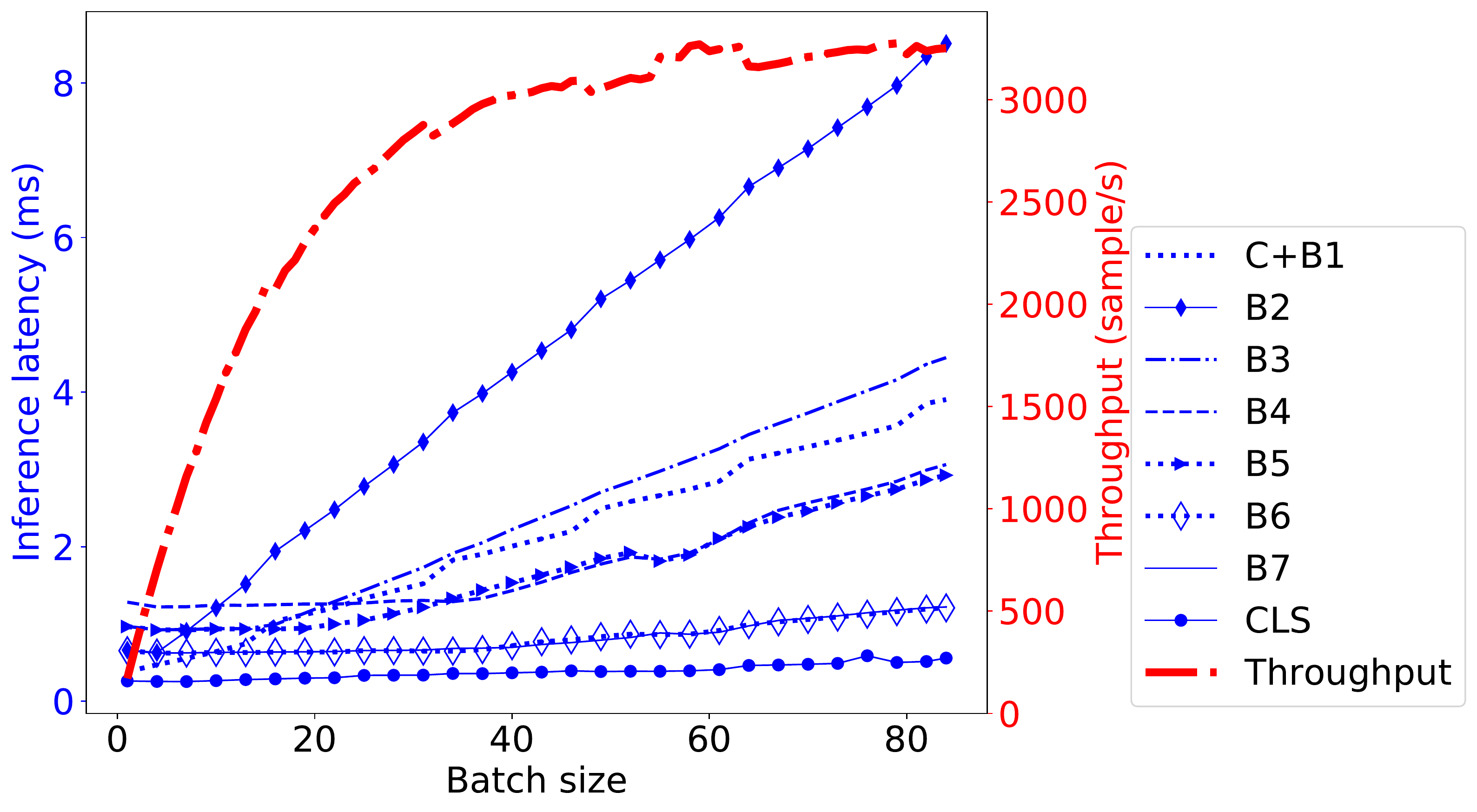}}
\label{l_b_mobilenetv2}
\caption{Profiling results for the sub-tasks in 3dssd (shown in Fig. \ref{l_b}. (a)) and mobilenet-v2 (shown in Fig. \ref{l_b}. (b)). Blue curves show the inference latency w.r.t. batch size for each sub-task (i.e., $F_n(\cdot)$), while red curves show the throughput of the entire inference task w.r.t. batch size. Implementing details will be introduced in the experiment part.}
\label{l_b} 
\vspace{-8pt}
\end{figure}
Besides, we have the edge server occupancy constraint
\begin{equation}
    s_{k+1} \geq s_k + F_n(b_{n,k}), \  \forall k \geq 1 \ \text{and} \ \forall n \in \{1, \dots, N\}, \label{C16}
\end{equation}
indicating that the next batch cannot be started until the current batch is completed. 
Notice that \eqref{C13} guarantees that for a typical batch $k$, there is at most one $\hat{n}$ among all $n \in \{1, \dots, N\}$ such that $b_{\hat{n},k} > 0$.
Therefore, if we define $F_n(0) = 0$, then \eqref{C16} can be satisfied for all $k \geq 1$ and $n \in \{1, \dots, N\}$.
Further, since the edge server has stronger computing capability than mobile devices, we assume that the edge computing latency when the batch size equals one is less than the minimum local computing latency, i.e., $F_n(1) \leq \frac{A_n}{f_{m,\text{max}}}$.

To ensure that for each user, the current sub-task will not be processed before the previous sub-task is completed, we have the following sub-task precedence constraint
\begin{align}
    t_{m,n} & \geq
    x_{m, n, 0} \underbrace{\left(t_{m,n-1} + \frac{B_{n-1}}{R_m^\text{d}} y^\text{d}_{m,{n-1}} + \frac{A_{n}}{f_m}\right)}_{\text{completing time if processed locally}} + x_{m, n, k} \underbrace{(s_k + F_n(b_{n,k}))}_{\text{completing time if offloaded}}, \  \forall k\geq 1, \label{C17}
\end{align}
where the binary variable $y^\text{d}_{m,n-1}$ is used to denote whether the 
output intermediate data of the $(n-1)$-th sub-task needs to be downloaded, and is defined as follows.
\begin{equation}
    y_{m,n-1}^\text{d} = \mathds{1} \{x_{m,n,0} - x_{m,n-1,0} > 0\}. \label{C18}
\end{equation}
In \eqref{C17}, the first term is the completing time of processing the sub-task locally, which consists of the completing time of the previous sub-task, the downloading latency of intermediate data (if needed), and the local computing latency.
While the second term is the completing time of offloading the sub-task, which equals the completing time of the batch that contains the sub-task.


\subsection{Problem Formulation}
\textcolor{black}{The user energy minimization problem under inference latency constraints is given by}
\begin{align}
    \underset{x_{m,n,k}, \, s_k, \, t_{m, n}, \, f_m}{\text{min}} \quad & \sum_{m=1}^M\sum_{n=1}^N \kappa_m A_n f_m^2 x_{m,n,0} + \frac{B_n}{R_m^\text{u}} p_m^\text{u} y_{m,n}^\text{u} + \frac{B_n}{R_m^\text{d}} p_m^\text{d} y_{m,n}^\text{d}
    \tag{P1}\label{P1}\\
    \text{s.t.} \qquad \quad \, \ & \eqref{C11}-\eqref{C18}, \nonumber \\
    & t_{m, N} \leq l_m, \label{C19}\\
    & f_{m,\text{min}} \leq f_m \leq f_{m,\text{max}}, \ t_{m,n} \geq 0, \ s_k \geq 0, \label{C20}\\
    & x_{m,0,0} = 1, t_{m,0} = 0. \label{C21}
\end{align}
The objective of \ref{P1} is the total user energy consumption, including the energy consumption of local computing, intermediate data uploading, and downloading.
\textcolor{black}{\eqref{C19} ensures that the latency constraint is satisfied, and \eqref{C21} means that the input data are ready for all users at $t=0$.}

Problem \ref{P1} is a mixed-integer nonlinear programming problem, where the nonlinearity is introduced by $F_n(\cdot)$.
Even if approximating $F_n(\cdot)$ by a linear function, mixed-integer programming problems are still NP-hard \cite{wolsey2007mixed}.
Besides, \ref{P1} contains an exploding solution space when the problem scales up.
For example, the number of possible solutions for $x_{m,n,k}$ is $2^{MN(MN+1)} \approx 2^{M^2 N^2}$, which makes commonly used heuristics (e.g., branch-and-bound) impractical for \ref{P1}.


\section{Simplified Problem and its Optimal Solution}
In this section, we introduce two simplifications to \ref{P1}: 1) all users have the same latency constraint, i.e., $l_m = l$, 2) the computing latency of edge server is not related to the batch size, i.e., $F_n(b_{n,k}) = F_n(1), \ \forall b_{n,k} \geq 1$. 

\begin{theorem}
\label{theorem1}
If simplification 1) and 2) hold, then \ref{P1} must have an optimal solution $x^*, s^*, t^*, f^*$ that satisfies: 

(1) \textcolor{black}{If the $n$-th sub-task of the $m$-th user is offloaded to the edge server, then all the following sub-tasks of the $m$-th user will be offloaded to the edge server (i.e., if $ x^*_{m,n,0} = 0$, then $x^*_{m,n',0}=0, \ \forall N \geq n' \geq n$).}

(2) The same sub-tasks that are offloaded to the edge server will be scheduled in the same batch (i.e., $b^*_{n,k}b^*_{n,k'}=0, \ \forall k, \ k' \geq 1  \ \text{and} \ k' \neq k$), and the batch starting time is given by 
\begin{equation}
\left\{
\begin{aligned}
    s^*_N &=  l - F_N(1), \\
    s^*_{N-1}   &=  s^*_N - F_{N-1}(1), \\
    &\vdots \\
    s^*_1  &= s^*_2 - F_1(1).
\end{aligned}
\right.
\label{s^*}
\end{equation}

(3) $f^*_m$ is the lowest device computing frequency that ensures the inference latency does not exceed the latency constraint $l$.
\end{theorem}

\begin{proof}
See Appendix \ref{appendix1}.
\end{proof}

Theorem \ref{theorem1} greatly reduces the complexity of solving \ref{P1}. In particular, Theorem \ref{theorem1}. (2) gives the optimal scheduling and decouples the offloading decisions of different users, and Theorem \ref{theorem1}. (1), Theorem \ref{theorem1}. (3) can be utilized to derive the optimal offloading decision for each user.
Based on Theorem \ref{theorem1}, we propose Alg. \ref{alg1} to derive the optimal offloading and scheduling policy.

\begin{algorithm}[!t]
\setlength{\abovecaptionskip}{2pt}
\setlength{\belowcaptionskip}{2pt}
    \caption{Traverse algorithm to find the optimal solution for simplified \ref{P1}}
    \label{alg1}
    \begin{algorithmic}[1]
    \STATE {Derive $s_k^*$ according to \eqref{s^*}}
    \FOR{$m \in \{1, \dots, M\}$}
        \FOR{$n \in \{0, \dots, N\}$}
            \STATE {Derive $f_{m,n}$ for the $m$-th user with partition point $n$ according to \eqref{f^*}}
            \STATE {Compute the energy consumption $E_{m,n}$ of the $m$-th user with partition point $n$}
        \ENDFOR
        \STATE {Choose the best partition point for user $m$, i.e., $n^* = \argmin \limits_n E_{m,n}$}
        \STATE {\textcolor{black}{Compute the corresponding $E^*_m = E_{m,n^*}$, $x^*_{m,n,k}, \forall n, k$, $t^*_{m,n}, \forall n$, and $f^*_m$}}
    \ENDFOR
    \RETURN {\textcolor{black}{$E^*=\sum \limits_m E^*_m$, $s_k^*$, $x^*_{m,n,k}$, $t^*_{m,n}$, and $f^*_m$}}
    \end{algorithmic}
\end{algorithm}

In Alg. \ref{alg1}, the optimal batch starting time $s^*_k$ is first derived (step 1).
Due to the fixed $s^*_k$, the offloading decisions of different users are no longer coupled.
According to Theorem \ref{theorem1}. (1), the whole inference task is partitioned into two parts, and the part before the partition point is locally processed while the rest is offloaded to the edge server.
Therefore, we can traverse all possible partition points for a typical user $m$.
For the partition point at the $n$-th sub-task, the latest completing time of the $n$-th sub-task is $s^*_{n+1} - \frac{B_n}{R^\text{u}_m}$.
According to Theorem \ref{theorem1}. (3), the corresponding local computing frequency is 
\begin{equation}
\label{f^*}
     f_{m,n}= \left\{
    \begin{array}{ll}
        \frac{\sum \limits_{i=0}^n A_{i}}{s^*_{n+1} -  \frac{B_n}{R^\text{u}_m}}  & \text{, if $f_{m,\text{min}} \leq \frac{\sum \limits_{i=0}^n A_{i}}{s^*_{n+1} -  \frac{B_n}{R^\text{u}_m}} \leq f_{m,\text{max}}$}, \\
        f_{m,\text{min}}  & \textcolor{black}{\text{, if $0 < \frac{\sum \limits_{i=0}^n A_{i}}{s^*_{n+1} -  \frac{B_n}{R^\text{u}_m}} < f_{m,\text{min}}$} ,}\\
        \text{does not exist} & \textcolor{black}{\text{, otherwise}} ,
	\end{array}
	\right.
\end{equation}
where $\sum \limits_{i=0}^n A_{i}$ is the total local computing workload up to the $n$-th sub-task (step 4).
Further, the corresponding user energy consumption can be derived by subustituing the offloading decision and local computing frequency into the objective of \ref{P1} (step 5).
Finally, we can compare the energy consumption of all possible partition points to derive the optimal partition point (i.e., $n^*$), 
and the corresponding user energy consumption, offloading and scheduling decision, sub-task completing time, and local computing frequency (steps 7-8).

Step 4 and step 5 iterate for $\mathcal{O}(MN)$ times. 
The complexity of step 4 is on the magnitude of $\mathcal{O}(1)$ if a lookup table of $\sum \limits_{i=0}^n A_{i}$ is established.
The complexity of step 5 is the same as step 4, and the complexity of step 7 and step 8 are both $\mathcal{O}(N)$.
Therefore, the complexity of Alg. \ref{alg1} is $\mathcal{O}(MN)$.

\section{Extension to Sophisticated Scenarios}
In this section, we first investigate the question: can the two simplifications (i.e., the same latency constraint, and constant edge inference latency for a typical sub-task) be moved? 
The answer is yes, but the optimality no longer holds.
Two different algorithms are proposed, and the experiment results show that these algorithms perform well.
Then, an online scenario that the arrivals of future tasks cannot be precisely predicted is investigated.
A reinforcement learning (RL) agent is proposed to schedule the offline algorithms according to the task arrivals and the state of the edge server.

\subsection{Realistic $F_n(\cdot)$}
As shown in Fig. \ref{l_b}, the inference latency of the edge server increases with the batch size, especially when the DNN is computational intensive or the batch size is large.
Therefore, directly applying Alg. \ref{alg1} may violate the latency constraint.
Therefore, Alg. \ref{IPSSA} is proposed, and the main idea is to deal with the worst case.

\begin{algorithm}[!t]
\setlength{\abovecaptionskip}{2pt}
\setlength{\belowcaptionskip}{2pt}
    \caption{Independent partitioning and same sub-task aggregating (IP-SSA)}
    \label{IPSSA}
    \begin{algorithmic}[1]
    \STATE {\textcolor{black}{Initialize $E^*=\infty$}}
    \FOR{$b = M, M-1, \dots, 1$}
        \STATE {Derive $s_k^*(b)$ by replacing $F_n(1)$ with $F_n(b)$ in \eqref{s^*}}
        \STATE {\textcolor{black}{Call Alg. \ref{alg1} to derive a feasible solution $E^*(b)$, $x^*_{m,n,k}(b)$, $t^*_{m,n}(b)$, and $f^*_m(b)$ for $s_k^*(b)$}}
        \STATE {Compute the actual maximum batch size $b_\text{max} = \max \limits_{n} b^*_{n,k}(b)$ among all sub-tasks}
        \IF{\textcolor{black}{$b_\text{max} \leq b$ and $E^*(b) < E^*$}}
            \STATE {\textcolor{black}{$E^* \gets E^*(b)$, $s_k^* \gets s_k^*(b)$, $x^*_{m,n,k} \gets x^*_{m,n,k}(b)$, $t^*_{m,n}\gets t^*_{m,n}(b)$, and  $f^*_m \gets f^*_m(b)$}}
        \ENDIF
    \ENDFOR
    \RETURN {\textcolor{black}{$E^*$, $s_k^*$, $x^*_{m,n,k}$, $t^*_{m,n}$, and $f^*_m$}}
    \end{algorithmic}
\end{algorithm}

Assuming $b$ is the maximum batch size among all sub-tasks.
A feasible solution can be derived by calling Alg. \ref{alg1} after replacing $F_n(1)$ with $F_n(b)$ in \eqref{s^*}, which ensures the inference latency constraint when the maximum batch size is $b$ \textcolor{black}{(steps 3-4)}.
Then, the actual maximum batch size of the derived solution can be calculated as $b_\text{max} = \max \limits_{n} b^*_{n,k}(b)$ \textcolor{black}{(step 5)}.
We compare the actual maximum batch size $b_\text{max}$ and the assumed batch size $b$.
$b_\text{max} \leq b$ means that the derived solution for $b$ is feasible and does not violate the inference latency constraint, otherwise it is infeasible \textcolor{black}{(steps 6-8)}.
\textcolor{black}{In the beginning of Alg. \ref{IPSSA}, let $b$ equal the number of users $M$, and so we have the worst case solution, since a longer edge computing latency decreases the probability of offloading and thus increases the user energy consumption.
After that, we can traverse $b$ to derive the feasible solution that has the minimum user energy consumption.}
In Alg. \ref{IPSSA}, the loop iterates $M$ times.
In each iteration, the complexity is dominated by Alg. \ref{alg1}, which is $\mathcal{O}(MN)$.
Therefore, the complexity of Alg. \ref{IPSSA} is $\mathcal{O}(M^2N)$
\textcolor{black}{\footnote{\textcolor{black}{
Alg. \ref{alg1} can be extended to the scenarios that the tasks have different arrival time $t_{m,0}$, by replacing the $s^*_{n+1} - \frac{B_n}{R^\text{u}_m}$ term in \eqref{f^*} with $s^*_{n+1} - \frac{B_n}{R^\text{u}_m}-t_{m,0}$.
Further, Alg. \ref{IPSSA} can also be extended by applying the extended version of Alg. \ref{alg1}.}}}.

Although Alg. \ref{IPSSA} can derive a feasible solution, the optimality cannot be guaranteed due to the following reasons.
On one hand, $b$ is the maximum batch size among all sub-tasks. 
For those sub-tasks with smaller batch sizes, the batch starting time given by Alg. \ref{IPSSA} is earlier than the optimal one, and thus mobile devices may consume more energy.
On the other hand, Theorem \ref{theorem1} no longer holds when $F_n(b)$ increases with $b$.
Therefore, Alg. \ref{alg1} is not optimal, as is Alg. \ref{IPSSA}. 

\subsection{Different Latency Constraints}
Based on Alg. \ref{IPSSA}, we further consider the scenario that users have different latency constraints. 
Without loss of generality, we assume that the latency constraints satisfy $l_1 \leq l_2 \dots \leq l_M $.
Since Theorem \ref{theorem1} no longer holds when the latency constraints are different, and the original problem \ref{P1} is NP-hard, we propose a divide and conquer method.

Intuitively, the tasks with similar latency constraints can be processed in one batch with little performance loss.
Therefore, we aim to divide all tasks into groups, and call Alg. \ref{IPSSA} to derive the offloading and scheduling policy for each group. 
Use $\mathcal{G}_1, \mathcal{G}_2, \dots, \mathcal{G}_g$ to denote the grouping policy, where $\mathcal{G}_i, i\in \{1, 2, \dots, g\}$ is a set of users.
\textcolor{black}{For each group, we set the latency constraint equal to the minimum one among users in the group, i.e., }
\begin{equation}
\label{tildel}
    \tilde{l}_i = \min \limits_{m\in \mathcal{G}_i} l_m.
\end{equation} 
With $\tilde{l}_i$, all users in $\mathcal{G}_i$ can satisfy their latency constraints.
Without loss of generality, we assume that $\tilde{l}_1 \leq \tilde{l}_2 \dots \leq \tilde{l}_g$.
It is also assumed that the edge processing of adjacent groups does not overlap, i.e.,
\begin{equation}
\label{nonoverlap}
    \tilde{l}_i + \sum \limits_{n=1}^N F_n(|\mathcal{G}_{i+1}|) \leq \tilde{l}_{i+1}, \ \forall i\in \{1,2,\dots, g-1\},
\end{equation}
where $|\mathcal{G}_{i+1}|$ is the number of users in group $\mathcal{G}_{i+1}$, and thus $\sum \limits_{n=1}^N F_n(|\mathcal{G}_{i+1}|)$ is the maximum length of the edge occupancy period of $\mathcal{G}_{i+1}$. 
\textcolor{black}{Although assumption \eqref{nonoverlap} restricts certain grouping policies, it can greatly reduce the complexity by forcing the tasks with similar latency constraints to be grouped together, without degrading the performance too much.}

Instead of solving the original problem \ref{P1}, a grouping problem is considered in order to provide practical solutions based on the IP-SSA algorithm (i.e., Alg. \ref{IPSSA}).
The problem is: under assumptions \eqref{tildel} and \eqref{nonoverlap}, how to form groups to minimize the user energy consumption?

\begin{theorem}
\label{subsequent_users}
There exists an optimal grouping $\mathcal{G}^*_1, \mathcal{G}^*_2, \dots, \mathcal{G}^*_g$, where each group contains and only contains several subsequent users, and ordered by the index of the group, i.e., $\mathcal{G}^*_i=\{m_i, m_i+1, \dots, m_{i+1}-1\}$ and $1=m_1\leq m_2\leq \dots \leq m_{g+1}=M+1$. 
\end{theorem}
\begin{proof}
See Appendix \ref{appendix2}.
\end{proof}

Theorem \ref{subsequent_users} confirms the intuition that the tasks with similar latency constraints should be grouped together. 
Further, the optimal grouping policy can be derived by dynamic programming based on Theorem \ref{subsequent_users}.
First, we introduce two auxiliary variables:
$S_{i,j}$ is used to denote the minimum user energy consumption for tasks $\{1,2,\dots,j\}$, where the first task in the last group is the $i$-th task.
$G_{i,j}$ is used to denote the output user energy consumption of calling Alg. \ref{IPSSA} for group $\mathcal{G}=\{i, i+1, \dots, j\}$ with $\tilde{l} = l_i$.
The dynamic programming algorithm for optimal grouping is shown in Alg. \ref{OG}.

\begin{algorithm}[!t]
\setlength{\abovecaptionskip}{2pt}
\setlength{\belowcaptionskip}{2pt}
    \caption{Dynamic programming algorithm for optimal grouping (OG)}
    \label{OG}
    \begin{algorithmic}[1]
    \STATE {Initialize $S_{1,1}=G_{1,1}$}
    \FOR {$i \in \{1, 2, \dots, M-1\}$}
        \FOR {$j \in \{i+1, i+2, \dots, M\}$}
                \STATE {$S_{i,j} = S_{i,i} - G_{i,i} + G_{i,j}$}
        \ENDFOR
        \STATE {$\mathcal{D}=\{i'|l_{i'}+\sum \limits_{n=0}^N F_n(i+1-i') \leq l_{i+1}\}$}
        \IF {$\mathcal{D} \neq \emptyset$}
            \STATE {$S_{i+1,i+1} = \min \limits_{i' \in \mathcal{D}} \{S_{i',i}\}+ G_{i+1, i+1}$}
        \ELSE 
            \STATE {$S_{i+1,i+1} = +\infty $}
        \ENDIF
    \ENDFOR
    \RETURN {$\min \limits_{i\in\{1, 2,\dots, M\}} S_{i,M}$} 
    \end{algorithmic}
\end{algorithm}

In Alg. \ref{OG}, we first initialize $S_{1,1}=G_{1,1}$. 
Since if there is only one task, the optimal grouping is trivial, and the optimal user energy consumption is given by calling Alg. \ref{IPSSA} for the task.
According to the definition, $S_{i,j}$ consists of two parts, the user energy consumption of the tasks in the last group (i.e., $G_{i, j}$) and the user energy consumption of all other groups (i.e., $S_{i,j} - G_{i, j}$).
Therefore, given $S_{i,i}$, step 4 can provide $S_{i,j}$, since the grouping policy that gives $S_{i,i}$ is the same as that of $S_{i,j}$ except the last group.
To derive $S_{i+1,i+1}$, we first need to derive the feasible region of the second last group, where the latency constraint of the second last group cannot be too close to that of the last group according to \eqref{nonoverlap}.
As a result, $\mathcal{D}$ is used to denote the set of all feasible indices of the first task in the second last group (step 6). 
If $\mathcal{D}$ is not empty, $S_{i+1,i+1}$ can be calculated as the optimal user energy consumption of tasks $\{1,2,\dots,i\}$ (i.e., $\min \limits_{i' \in \mathcal{D}} \{S_{i',i}\}$) plus the energy consumption of the $(i+1)$-th task (step 8).
Otherwise, it means that the latency constraint of the $(i+1)$-th task is too close to that of the previous tasks, and thus form a new group for the $(i+1)$-th task is infeasible (step 10).
Finally, the optimal grouping can be derived by traversing all possible indices of the first task in the last group (step 13).

\textcolor{black}{The complexity of calling Alg. \ref{IPSSA} to derive $G_{i,j}$ for a typical pair of $(i, j)$ is $\mathcal{O}\left((j-i+1)^2N\right)$, which is no higher than $\mathcal{O}(M^2N)$.}
Thus, the complexity of deriving $G_{i,j}$ for all $(i, j)$ is $\mathcal{O}(M^4N)$.
Step 4 iterates for $\mathcal{O}(M^2)$ times, and the complexity of each iteration is $\mathcal{O}(1)$.
Steps 6-11 iterate for $\mathcal{O}(M)$ times.
The complexity of step 6 is $\mathcal{O}(M)$, and so is step 8.
As a result, the complexity of Alg. \ref{OG} is $\mathcal{O}(M^4N)$.

\subsection{Online Solution}
In real co-inference scenarios, future task arrivals may not be precisely predicted.
Therefore, we need to design online algorithm to make decisions on-the-fly.
In general, the online algorithm seeks the balance between \emph{serving arrived tasks} and \emph{reserving resources for future tasks}, which is affected by the following two trade-offs.
First, there is a trade-off between waiting latency before starting a batch and the batch size.
Since if the edge server chooses to wait, the batch size may be increased, which provides a higher inference throughput.
Second, there is a trade-off between the processing time of current batch and the idle period of the edge server.
Since if we enlarge the processing time of the current batch, the user energy consumption can be reduced.
However, the edge server will be occupied for a longer period, which sacrifices the capability of serving future tasks.
In the following part, we first model it via Markov decision process (MDP), and then propose a reinforcement learning (RL) agent to solve it.

Consider a slotted time system $t=1,2,\dots$, with time slot length $T$.
There are $M$ users in the system, and the task arrival of each user follows a random process $\mathcal{A}_m$.
The latency constraint of a task cannot be known unless it has arrived.
\textcolor{black}{We assumed that each user has a task buffer.
If a new task arrives before the latency constraint of the last task, it will be kept in the task buffer.}
The key components of the MDP are as follows.
\begin{itemize}
    \item {State:} 
    We use $\Vec{s}_t \triangleq [\Vec{l}_t, o_t]$ to denote the system state, where $\Vec{l}_t=[l_{1,t},\dots,l_{M,t}]$ is the latency constraints of all users in the $t$-th time slot ($l_{i,t}=0$ means the $i$-th user does not have  task to process), and $o_t$ is the busy period of the edge server (will be introduced in the state transition part).
    \item {Action:} 
    The two-dimensional action space is $\Vec{a}_t \triangleq [c_t, l_\text{th}]$, where $c_t\in \{0, 1, 2\}$ denotes doing nothing, letting the users to locally process their tasks, or calling OG, respectively. 
    \textcolor{black}{Further, $l_\text{th}$ is used to control the busy period of the edge server.}
    If the action is to call OG, i.e., $c_t=2$, for the tasks that have larger latency constraints (i.e., $l_{i,t}\geq l_\text{th}$), we force them to be completed before $l_\text{th}$ to reduce the busy period of the edge server. The two elements of the action vector are used to balance the two aforementioned trade-offs, respectively.
    \item {State transition:} 
    When the action is to do nothing,
    all non-trivial elements in the state vector are reduced by the time slot length, i.e., $l_{m,t+1} = l_{m,t} - T, \text{ if } l_{m,t}>0$ and $o_{t+1} = o_t - T, \text{ if } o_t>0$.
    If the action is to locally process or call OG, then $\Vec{l}_{t+1}=\vec{0}$ since all tasks will be processed.
    Besides, if calling OG, the busy period of edge server is $o_{t}=\tilde{l}_g$, where $\tilde{l}_g$ is the latency constraint of the last group in the optimal grouping policy given by OG.
    \item {Reward:} 
    The reward is defined as $r_t \triangleq -E(\Vec{s}_t, a_t)-C(\Vec{l}_t)$, where $E(\cdot)$ denotes the user energy consumption, and $C(\Vec{l}_t)$ denotes the cost.
    If the agent takes $c_t=0$ for a while, some tasks may have urgent latency constraints.
    Since the minimum latency of user $m$ to locally process the task is $l^\text{cp}_{m}(f_{m,\text{max}})=\frac{\sum_n A_n}{f_{m,\text{max}}}$.
    In order to ensure that those urgent tasks do not violate their delay constraints, 
    we assume that the user will locally process the task if the latency constraint cannot be satisfied in the next slot, i.e.,  $l_{m,t+1} \overset{c_t=0}= l_{m,t} - T < l^\text{cp}_{m}(f_{m,\text{max}})$.
    Therefore, we have $C(\Vec{l}_t) = \sum \limits_{m=1,\dots,M} e^\text{cp}_{m}(f_{m,\text{max}}) \mathds{1}\{(l_{m,t} - T) < l^\text{cp}_{m}(f_{m,\text{max}})\} $.
\end{itemize}

Notice that the state transition is complicated, since it is related to both the action and the output of Alg. \ref{OG}.
Therefore, theoretically solving this MDP is hard, we use RL to learn when and how to call OG.
Specifically, we use deep deterministic policy gradient (DDPG) \cite{lillicrap2016continuous}, since $l_\text{th}$ is continuous and DDPG is suitable for continuous control \textcolor{black}{\footnote{\textcolor{black}{
To derive $c_t$, we discretize the continuous output of the DDPG agent in equal width, and train the DDPG agent via the standard method provided in \cite{lillicrap2016continuous}.
Although in multiple GPUs or multiple servers scenarios, the complicated discrete-continuous hybrid action space may increase the difficulty of RL training, some advanced RL agents designed for hybrid action space can be applied \cite{neunert2020continuous, 10.5555/3367243.3367356}.
}}}.

\begin{figure}[!t] 
\setlength{\abovecaptionskip}{2pt}
\setlength{\belowcaptionskip}{2pt}
\centering 
\includegraphics[width=0.6\linewidth]{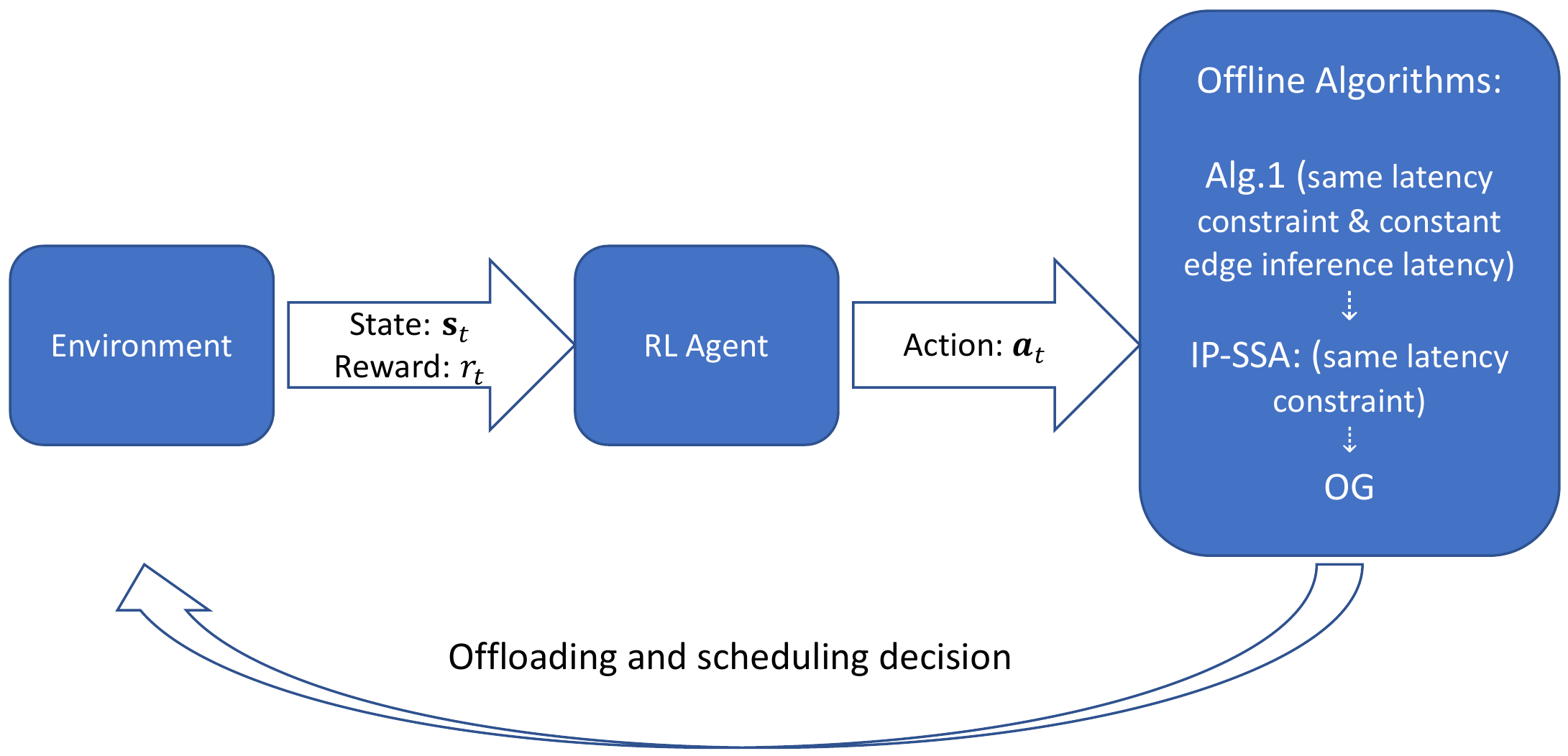} 
\caption{Illustration of the overall framework.} 
\label{rl}
\vspace{-8pt}
\end{figure}

In summary, the overall framework is shown in Fig. \ref{rl}.
For offline scenarios that all tasks have arrived and the latency constraints are known, Alg. 1, IP-SSA, and OG are proposed. 
Specifically, Alg. 1 is proven to be optimal when all tasks have the same latency constraint and the edge inference latency is not related to the batch size.
Then, Alg. 1 is extended to IP-SSA by finding the maximum feasible batch size in a greedy manner, for the scenario with the same latency constraint and increasing edge inference latency.
Considering different latency constraints, IP-SSA is further extended to OG by grouping the tasks with similar latency constraints.
When future task arrivals cannot be precisely predicted, a DDPG agent is trained to call OG, according to the latency constraints of arrived tasks, previous reward, and the occupancy of the edge server.
The agent can adaptively seek the balance between serving arrived tasks and reserving resources for future tasks via the proposed two-dimensional action space.

\section{Experiment Results}
\subsection{Experiment Setup}
We evaluate the proposed algorithms on two different DNNs, mobilenet-v2 \cite{sandler2018mobilenetv2} and 
3dssd \cite{yang20203dssd}.
The architecture and sub-task partitioning are shown in Fig. \ref{arch}.
We use the standard pretrained mobilenet-v2 model (on ImageNet data set) from pytorch \cite{paszke2019pytorch}, and train a 3dssd model (on KITTI data set) based on an open source LiDAR-based 3D object detection project \cite{openpcdet2020}.
The DNNs are profiled with different batch sizes on an NVIDIA RTX3090 \cite{3090}.
We use all data samples from test data set for inference, and record average inference latency for each sub-task.
The profiling results are shown in Fig. \ref{l_b}.

\subsection{Experimental Settings}
We assume that the users are uniformly distributed in a circular area of radius $R$, with an edge server located at the center.
The transmission rate is assumed to reach the Shannon capacity, i.e., ${R_m^\text{u}}=W_m \text{log}_2(1+\frac{\hat{p}^\text{u}_m h^2_m}{W_m N_0})$.
Here, $W_m$ denotes the wireless bandwidth, $h_m$ denotes the channel gain, $N_0$ is the noise power density, and $\hat{p}^\text{u}_m$ is the transmit power, which is typically much smaller than the power consumption of the transmitter, i.e., $p_m^\text{u}$.
The path loss is $128.1 + 37.6 \text{log}_{10} d$ ($d$ is the distance,
in km), and the standard deviation of shadow fading is 8 dB.

In order to derive the local computing energy and latency, the values of $f_{m,\text{min}}, f_{m,\text{max}}, \kappa_m$ are required.
However, these parameters are related to the computing hardware, and mobile users can have different types of computing hardware (e.g., CPU or GPU).
Therefore, we choose an alternative way.
The \emph{energy efficiency} is defined as the computation workload that can be completed with unit computing energy, i.e. $\mathcal{E}(f) \triangleq \frac{A}{e^\text{cp}(f)}$.
Note that both $\mathcal{E}$ and $e^\text{cp}$ are functions of the computing frequency $f$.
The energy efficiency of different computing hardware at the maximum computing frequency is reported in \cite{guo2019dl}, and the edge computing energy consumption can be estimated by multiplying the inference latency and GPU power consumption.
Therefore, we have the user energy consumption of the $n$-th sub-task as follows
\begin{equation}
\label{efmax}
    e^\text{cp}_{m,n}(f_{m,\text{max}}) = \frac{A_n}{\mathcal{E}_m (f_{m,\text{max}})} = \frac{\mathcal{E} _\text{e}(f_{\text{e},\text{max}})}{\mathcal{E}_m (f_{m,\text{max}})} e^\text{cp}_{n,\text{e}}(f_{\text{e},\text{max}}) = \frac{\mathcal{E} _\text{e}(f_{\text{e},\text{max}})}{\mathcal{E}_m (f_{m,\text{max}})} F_n(1)P_\text{e},
\end{equation}
where the subscript e means the edge server, and $P_e$ is the power consumption of the GPU.
Further, due to the different numbers of GPU cores and transistors, different series of GPUs may have different computing capabilities.
Nevertheless, the energy efficiency of different GPUs is almost the same for the same microarchitecture, e.g., GTX 1650 Super and GTX 1660 Ti with Turing microarchitecture \cite{guo2019dl}.
As a result, we define a parameter $\alpha$ as the ratio of local inference latency and edge inference latency, both at the maximum frequency
\begin{equation}
\label{alpha}
    \alpha_m \triangleq \frac{\frac{A_n}{f_{m,\text{max}}}}{F_n(1)} = \frac{l^\text{cp}_{m,n}(f_{m,\text{max}})}{F_n(1)}.
\end{equation}
Here $\alpha_m$ characterizes the local computing capability of user $m$. Then combining \eqref{alpha} with \eqref{tcp} and \eqref{ecp}, we have
\begin{equation}
\label{ef}
    e^\text{cp}_{m,n}(f_m) = \kappa_m A_n f^2_{m,\text{max}} \left( \frac{f_m}{f_{m,\text{max}}} \right)^2 = e^\text{cp}_{m,n}(f_{m,\text{max}}) \left( \frac{l^\text{cp}_{m,n}(f_{m,\text{max}})}{l^\text{cp}_{m,n}(f_m)} \right)^2
    = \frac{e^\text{cp}_{m,n}(f_{m,\text{max}}) \alpha_m ^2 F_n(1)^2}{l^\text{cp}_{m,n}(f_m)^2}.
\end{equation}
Substituting \eqref{efmax} into \ref{ef}, we can derive the user energy consumption at any $f_m$.

For 3dssd, we assume that mobile devices use GPUs for local computing, while CPUs are used for mobilenet-v2, respectively.
The default parameters are listed in Table \ref{parameter_offline}.

\begin{table}[!t]
\setlength{\abovecaptionskip}{2pt}
\setlength{\belowcaptionskip}{2pt}
\caption{System Parameters of the Offline Setting}
\label{parameter_offline}
\begin{center}
\begin{tabular}{|c |c |c |c|}
\hline
\textbf{Parameter} & \textbf{Value} & \textbf{Parameter} & \textbf{Value}\\
\hline
$R$ & 100 m & $\alpha_m$ & 1   \\
$W_m$ & 1 MHz & $N_0$ & -174 dBm/Hz \\
$\hat{p}^\text{u}_m$ & 0.05 W & $\mathcal{E} _\text{e}(f_{\text{e},\text{max}})$ & 48.75 Gop/W  \\
$p^\text{u}_m$ & 1 W & $\mathcal{E}_m(f_{m,\text{max}})$ for mobile GPU & 48.75 Gop/W  \\
$P_e$ & 300 W & $\mathcal{E}_m(f_{m,\text{max}})$ for mobile CPU & 0.3415 Gop/W \\
\hline
\end{tabular}
\end{center}
\vspace{-15pt}
\end{table}

\subsection{Offline Setting}
In this section, the proposed IP-SSA algorithm is evaluated, where all inference tasks have arrived and the latency constraints are the same.
We set $l_m=250$ ms for 3dssd, and $l_m=50$ ms for mobilenet-v2, respectively.
IP-SSA is compared with following benchmarks:
\begin{itemize}
    \item Local computing (LC): All users locally process the inference tasks.
    \item Offloading with processing sharing (PS): \textcolor{black}{All users evenly share the computing resources of the edge server, i.e., the edge computing latency of the $n$-th sub-task becomes $MF_n(1)$.
    Each user independently traverses all possible partition points of its DNN, and derives the partition point and the corresponding local computing frequency (similar to \eqref{f^*}) that minimize the user energy consumption.}
    \item Offloading with first-in-first-out (FIFO): \textcolor{black}{The edge server processes the offloaded sub-tasks in a FIFO manner. 
    The users are sorted by their transmission rate in descending order, and each user traverses all possible partition points of its DNN to derive the partition point that minimizes the user energy consumption. 
    Here, once a user offloads, the corresponding time period of edge processing is occupied and cannot be used by following users.
    We set $f_m=f_{m,\text{max}}$ to allow the edge server to process the most sub-tasks.
    }
    \item IP-SSA with no DNN partitioning (IP-SSA-NP): The whole DNN inference task is treated as one sub-task, and use IP-SSA to derive offloading and scheduling decisions.
\end{itemize}
Note that PS and FIFO only use simple scheduling policies, and thus can be used to evaluate the performance improvement of batch processing.
The improvement brought by DNN partitioning can be evaluated by comparing IP-SSA to IP-SSA-NP.

\begin{figure}[!t]
\setlength{\abovecaptionskip}{2pt}
\setlength{\belowcaptionskip}{2pt}
\centering
\subfloat[]{
\includegraphics[width=0.45\linewidth]{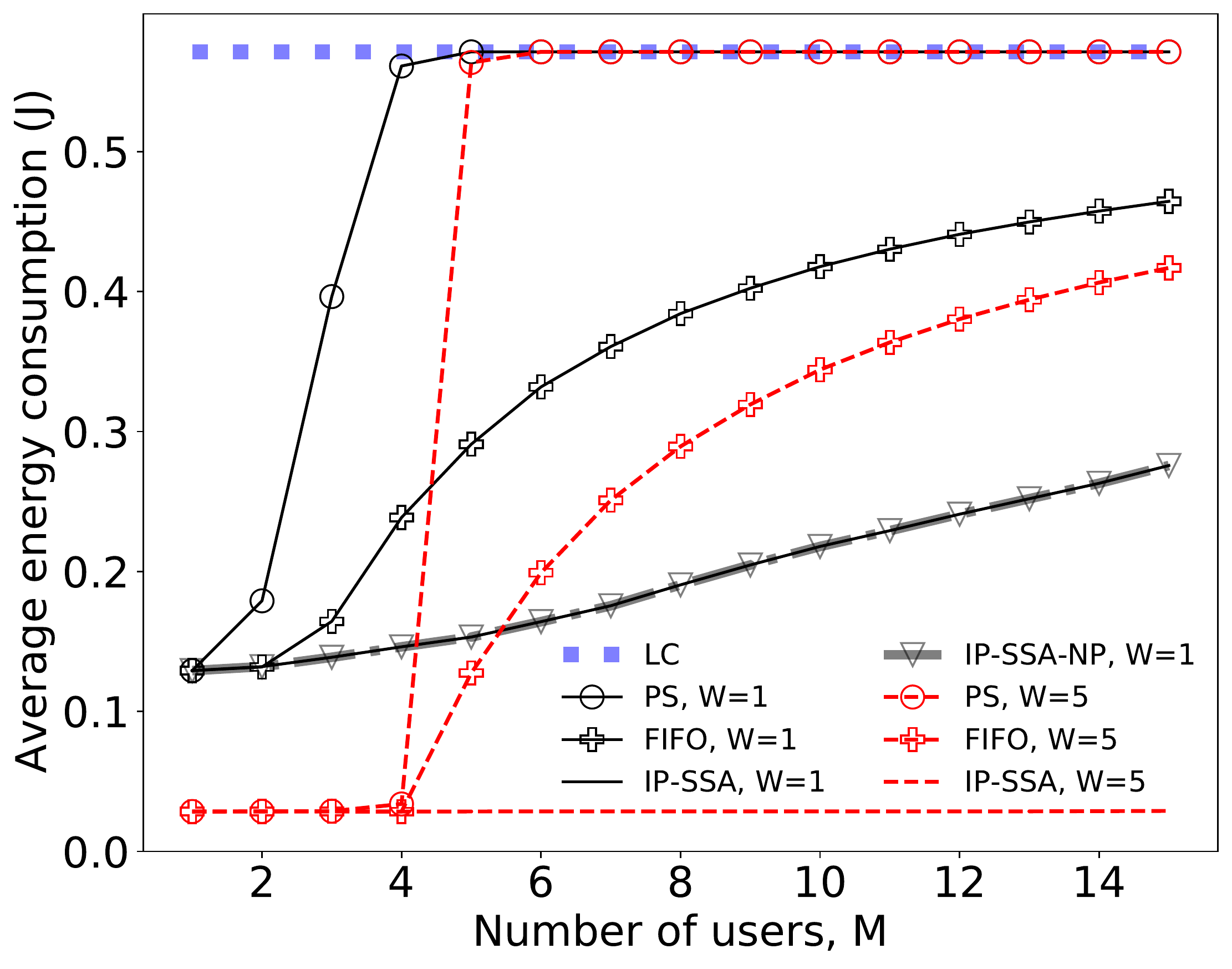}}
\label{o_b_3dssd}
\hfill
\subfloat[]{
\includegraphics[width=0.45\linewidth]{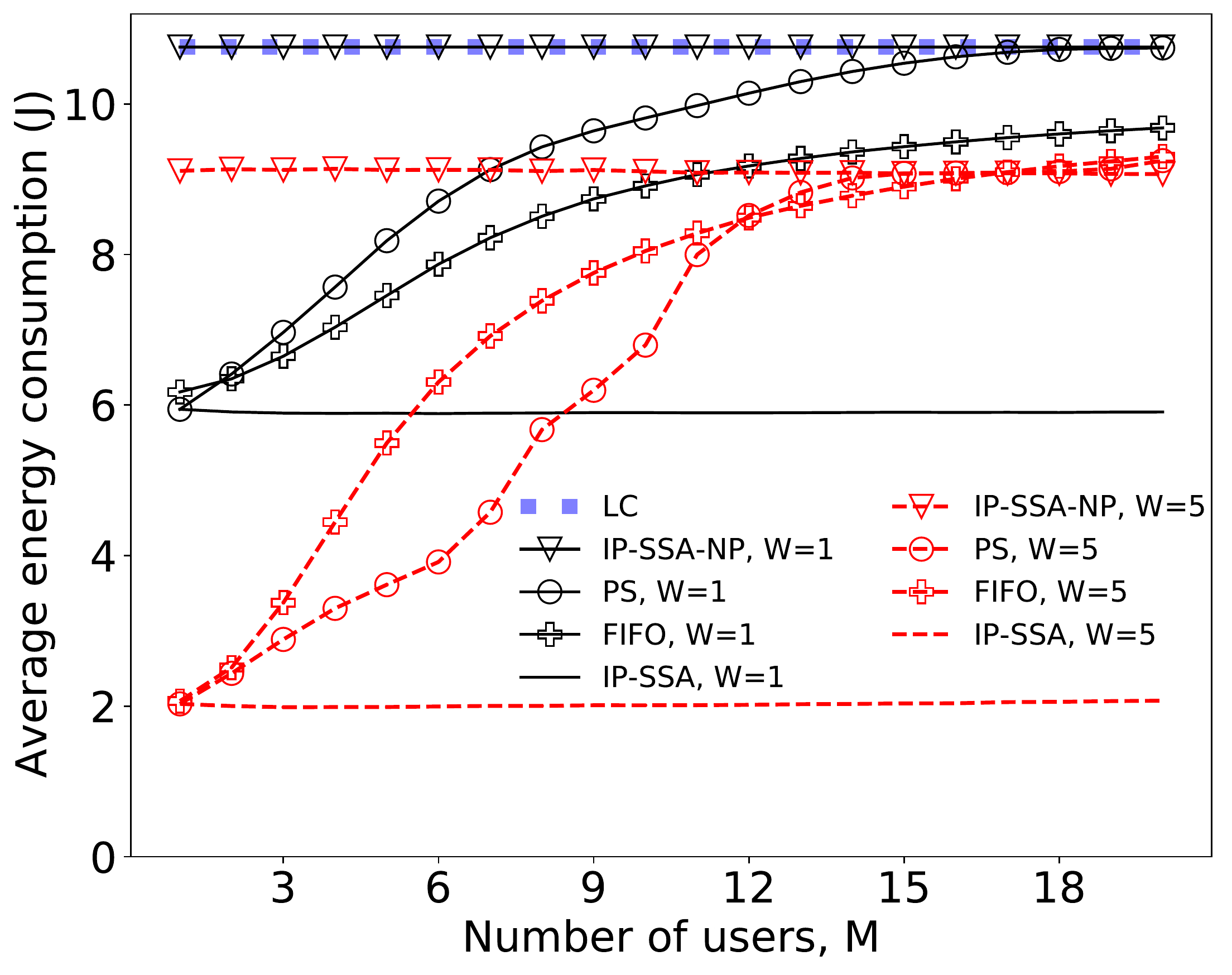}}
\label{o_b_mobilenet}
\caption{Average energy consumption per user v.s. the number of users under different wireless bandwidth (i.e., $W$, in MHz). Fig. \ref{o_b} (a) and (b) show the results for 3dssd and mobilenet-v2, respectively.}
\label{o_b} 
\vspace{-8pt}
\end{figure}

Fig. \ref{o_b} (a) and Fig. \ref{o_b} (b) show the average energy consumption per user w.r.t. the number of users for 3dssd and mobilenet-v2, respectively.
First, IP-SSA, IP-SSA-NP, PS, and FIFO consume less user energy with more wireless bandwidth (i.e., $W$, in MHz), since the transmission latency can be reduced.
This result suggests that the feature compression techniques \cite{ko2018edge,shi2019improving,shao2020communication,jankowski2020joint,shao2021learning} are beneficial for multi-user co-inference.
When the wireless bandwidth is the same, IP-SSA outperforms PS and FIFO, especially when the number of users is large. 
For 3dssd, IP-SSA reduces up to 40.6\% and 51.7\% user energy consumption compared to FIFO and PS when $W=1$ MHz and $M=15$, and reduces up to 93.1\% and 94.9\% user energy consumption compared to FIFO and PS when $W=5$ MHz and $M=15$.
The reason is that IP-SSA fully utilizes the computing resources of the edge server via batch processing, and more sub-tasks can be offloaded.
Comparing Fig. \ref{o_b} (a) and Fig. \ref{o_b} (b), we notice that the performance is closely related to the DNN architectures.
For 3dssd, IP-SSA-NP performs the same as IP-SSA, since the intermediate data for 3dssd is larger than the input data.
While for mobilenet-v2, DNN partitioning greatly reduces the user energy consumption.
For mobilenet-v2 with $W=1$ MHz, IP-SSA-NP cannot utilize the edge server, and thus performs the same as LC.
It is also shown that for mobilenet-v2, the performance of IP-SSA is not sensitive to the number of users.
The reason is that mobilenet-v2 is a light-weight DNN, and the edge inference latency is not sensitive to the batch size as shown in Fig. \ref{l_b}.
However, for 3dssd with $W=1$ MHz, IP-SSA consumes more energy per user when the number of user is large.
Since edge inference latency increases with the batch size, the batch starting time $s^*_k$ decreases, and thus fewer users can complete uploading the intermediate data before $s^*_k$, and more sub-tasks need to be locally processed.
Further, according to the result of IP-SSA for 3dssd with $W=5$ MHz, this phenomenon is less obvious when the users have more bandwidth, since the transmission latency can be reduced.

\begin{figure}[!t]
\setlength{\abovecaptionskip}{2pt}
\setlength{\belowcaptionskip}{2pt}
\centering
\subfloat[]{
\includegraphics[width=0.45\linewidth]{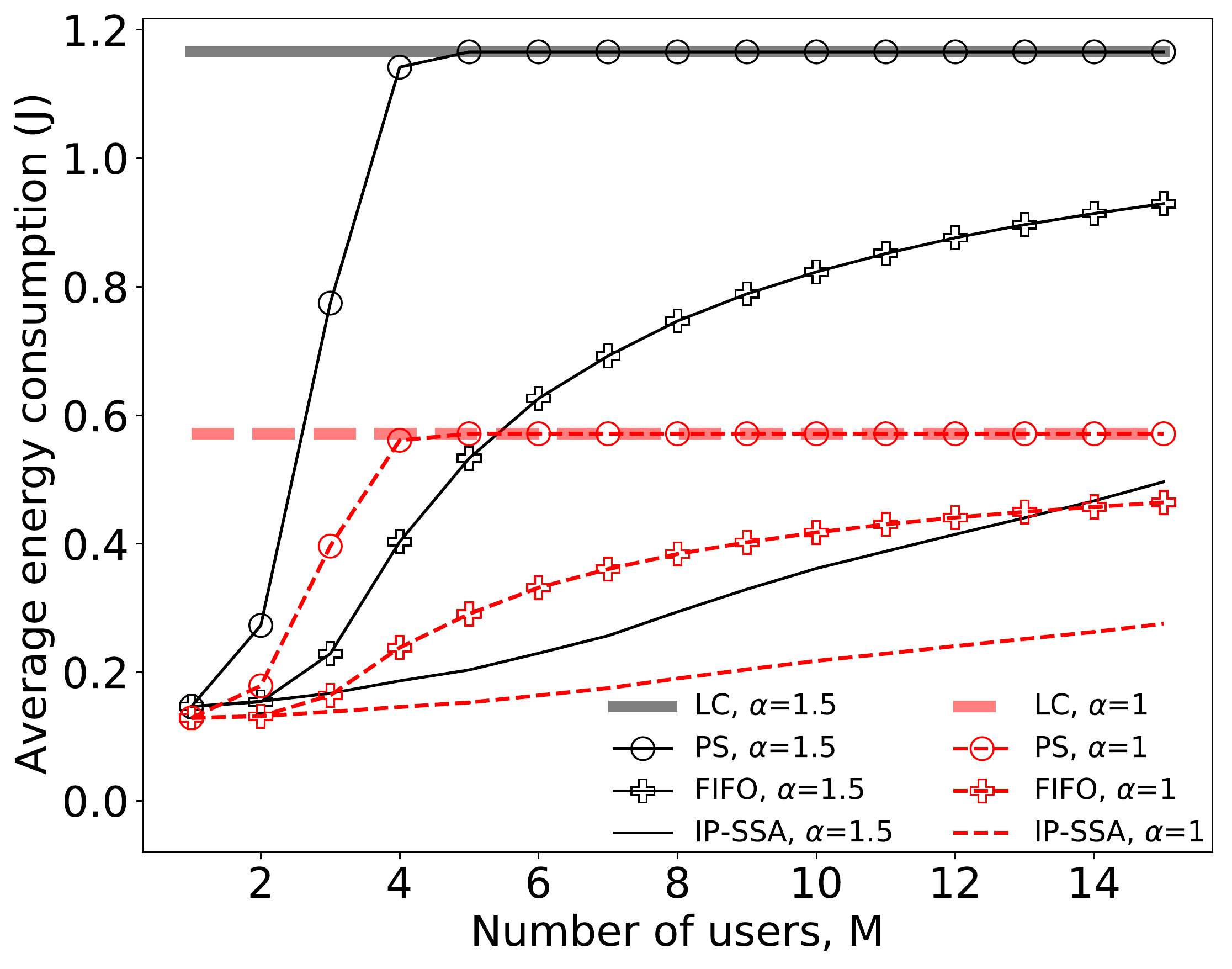}}
\label{o_r_3dssd}
\hfill
\subfloat[]{
\includegraphics[width=0.45\linewidth]{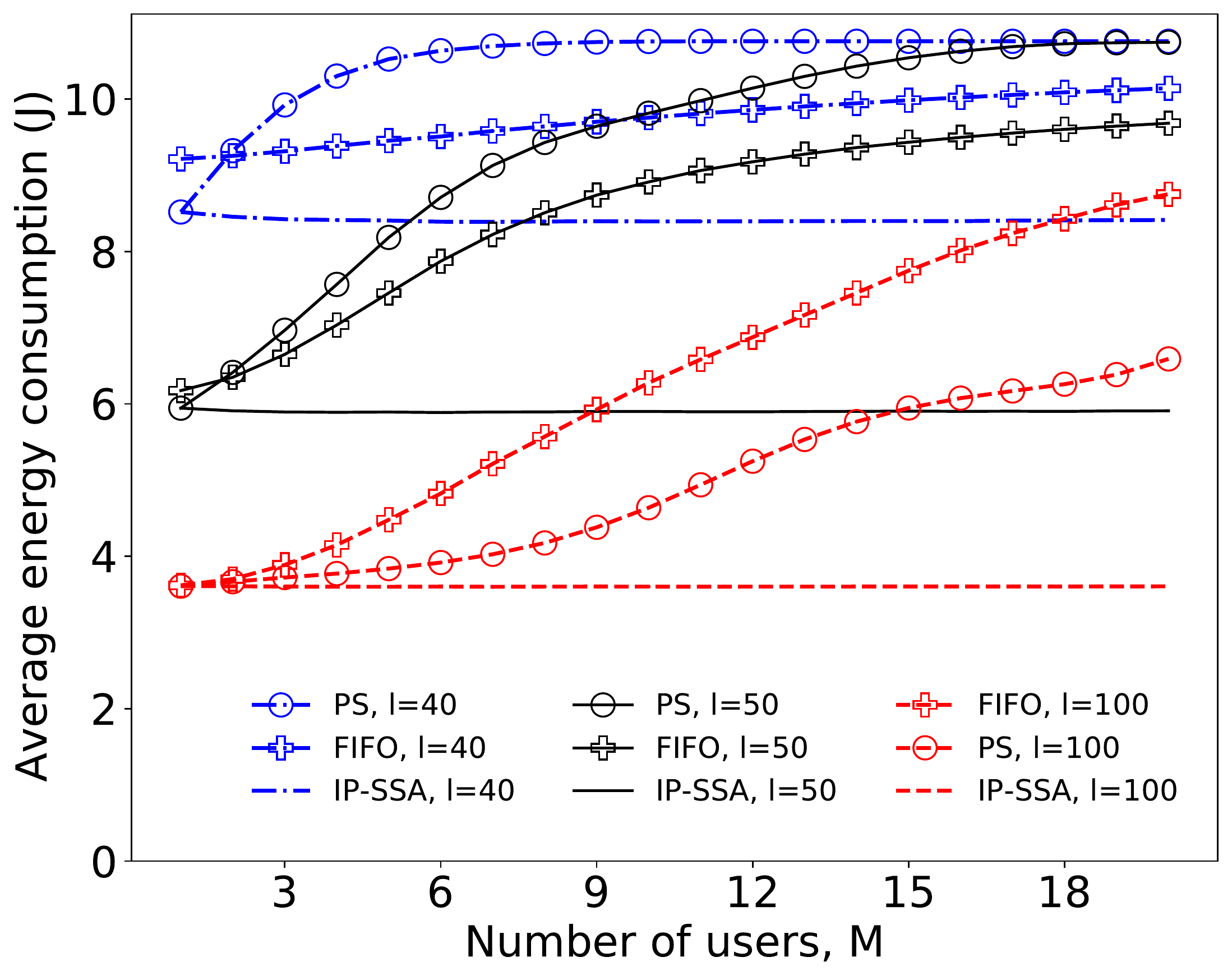}}
\label{o_l_mobilenet}
\caption{(a) shows the average energy consumption per user v.s. the number of users for 3dssd, under different mobile GPU computing capability parameters (i.e., $\alpha$). (b) shows the average energy consumption per user v.s. the number of users for mobilenet-v2, under different latency constraints (i.e., $l$).}
\label{ratio_latency} 
\vspace{-8pt}
\end{figure}

We evaluate the performance of different mobile GPU computing capability parameters $\alpha$ for 3dssd, and the results are reported in Fig. \ref{ratio_latency} (a).
According to \eqref{alpha}, a larger $\alpha$ means that the mobile GPU has weaker computing capability, and thus consumes more energy to complete the inference task within the same latency.
Fig. \ref{ratio_latency} (a) shows that when the number of users is small, the performance for different $\alpha$ is similar.
While the performance gap between different $\alpha$ becomes larger for more users, since the computing capability of edge server is fixed and thus more inference task should be processed locally.
Fig. \ref{ratio_latency} (a) also reveals two possible ways to reduce the user energy consumption.
On one hand, mobile GPUs with smaller $\alpha$ can be deployed to reduce user energy consumption, at the cost of more expensive mobile GPUs.
On the other hand, deploying more GPUs on edge server can also reduce the energy consumption per user by reducing the number of users served by each GPU.

The results of different inference task latency constraints for mobilenet-v2 are shown in Fig. \ref{ratio_latency} (b).
It is shown that the user energy consumption is sensitive when the latency constraint is low.
For IP-SSA with $M=10$, when the latency constraint is reduced from 100 ms to 50 ms, the average energy consumption increases by 2.57 J, and a further 10 ms latency constraint reduction needs 2.34 J.
\textcolor{black}{The results of the average batch size for each sub-task are shown in Table \ref{average_batch_size}.
As shown by Table III, the average batch size for the sub-tasks in the front part of mobilenet-v2 is smaller than that of the sub-tasks in the rear part, and the average batch size increases with the latency constraint.
These findings are consistent with Theorem 1 and the proposed IP-SSA algorithm.}
Moreover, the results of the distribution of user energy consumption are shown in Fig. \ref{e_d}.
The overlapped areas in the left-hand side bars in Fig. \ref{e_d} (a) and Fig. \ref{e_d} (b) show that for the FIFO policy, the high-priority users can have similar performance as IP-SSA by offloading more sub-tasks.
However, the red bars on the right show that in order to serve these high-priority users, the FIFO policy sacrifices other users that can only conduct the inference task locally with high energy consumption.
On the other hand, PS ensures fairness among users, since the user energy consumption is similar.
According to Fig. \ref{ratio_latency} (b), although PS performs better than FIFO when $l=100$ ms, such fairness can greatly increase the user energy consumption when the latency constraint is low.
As shown in Fig. \ref{e_d} (a), the scarce edge computing resource shared by each user may not ensure the stringent latency constraint, and thus more inference tasks need to be processed locally.
In contrast, the proposed IP-SSA can ensure both fairness and efficiency via batch processing.

\begin{figure}[!t]
\setlength{\abovecaptionskip}{2pt}
\setlength{\belowcaptionskip}{2pt}
\centering
\subfloat[]{
\includegraphics[width=0.42\linewidth]{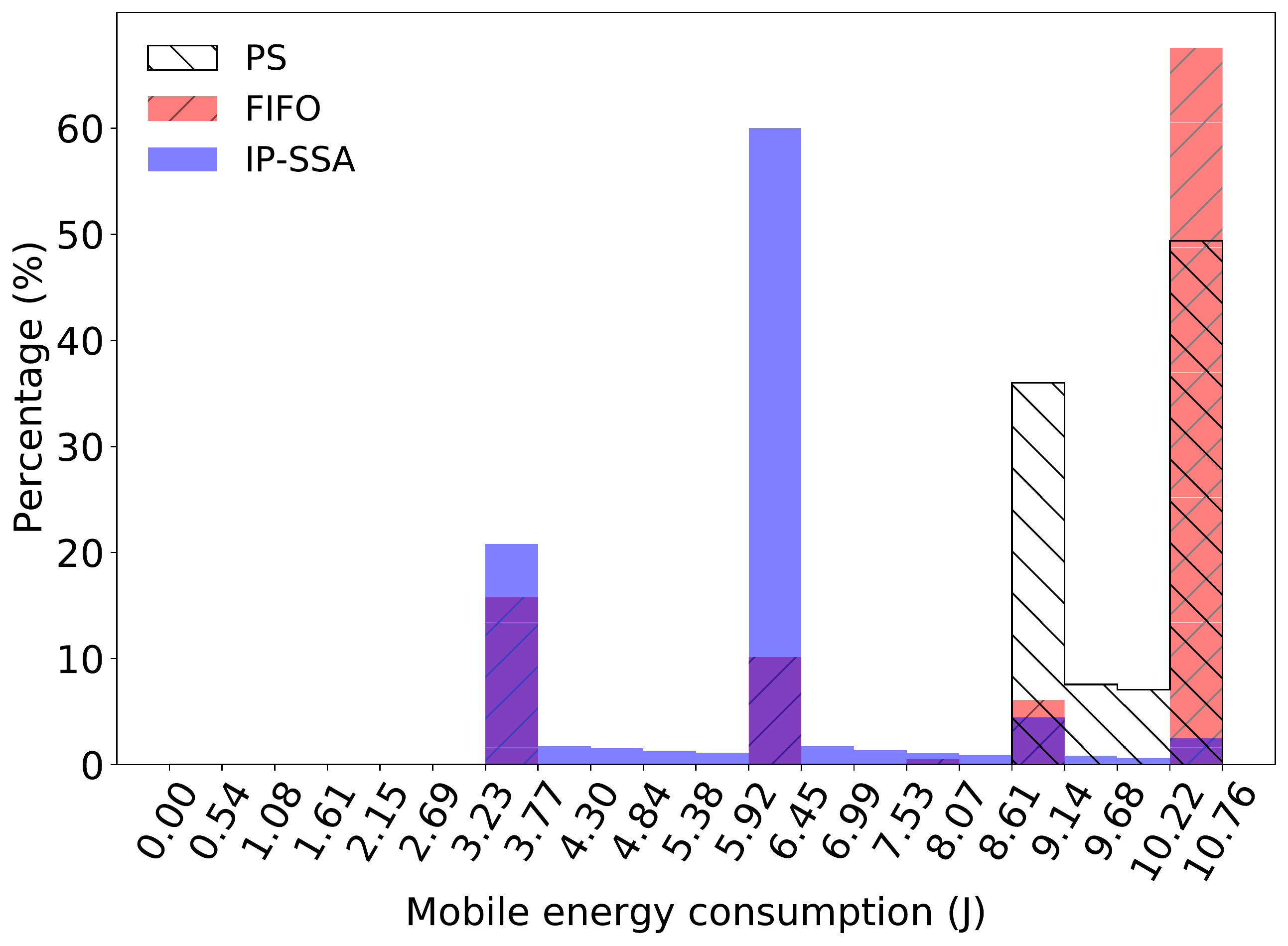}}
\label{e_d_50}
\hfill
\subfloat[]{
\includegraphics[width=0.42\linewidth]{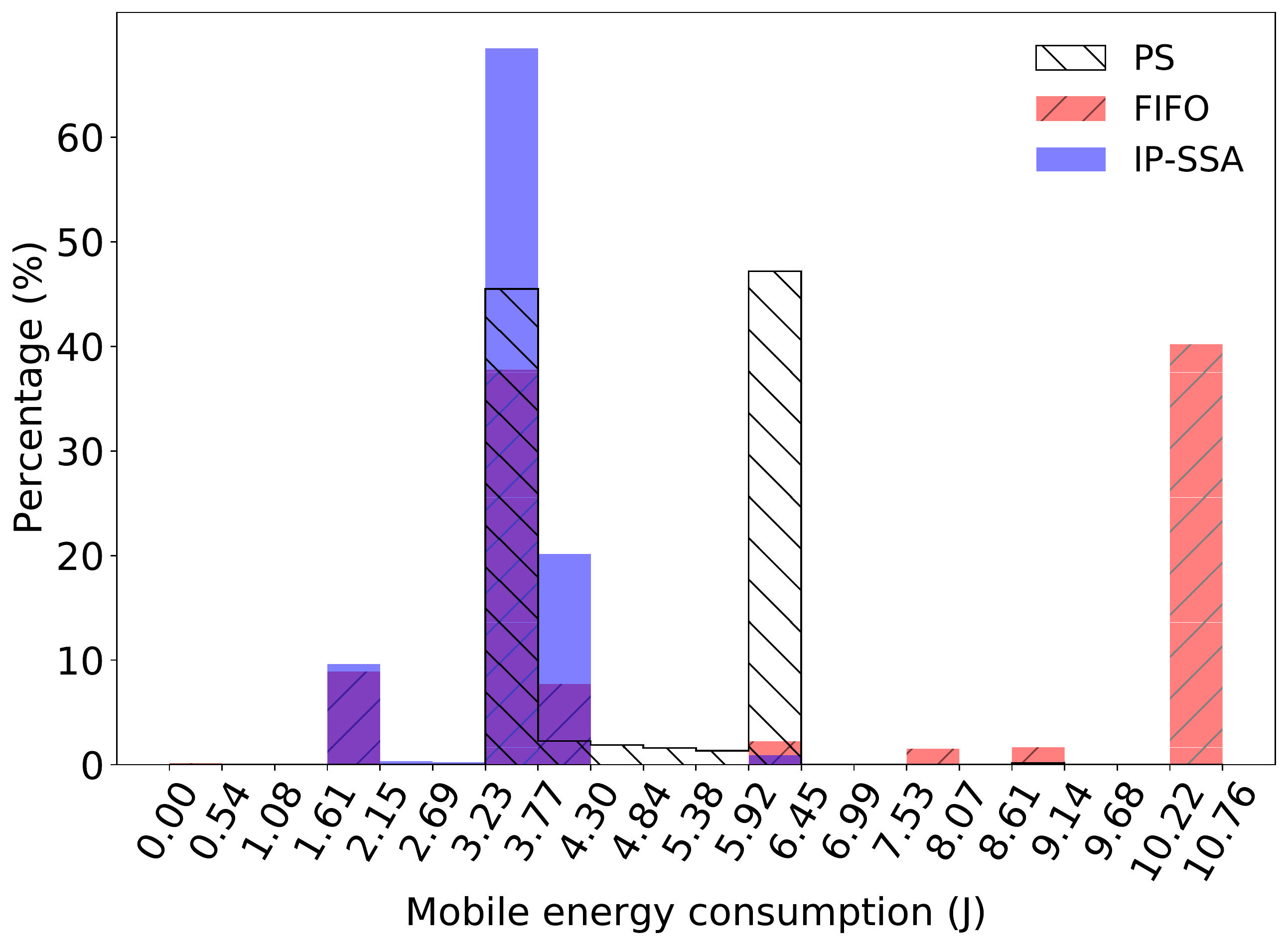}}
\label{e_d_100}
\caption{User energy consumption distributions when $M=10$. Fig. \ref{e_d} (a) is for $l=50$ ms, and Fig. \ref{e_d} (b) is for $l=100$ ms.}
\label{e_d} 
\vspace{-8pt}
\end{figure}

\begin{table}[!t]
\setlength{\abovecaptionskip}{2pt}
\setlength{\belowcaptionskip}{2pt}
\caption{\textcolor{black}{Average Batch Size of Each Sub-task for Mobilenet-v2 when $M=10$}}
\label{average_batch_size}
\begin{center}
\begin{tabular}{c c c c c c c c c }
\hline
& C+B1 & B2 & B3 & B4 & B5 & B6 & B7 & CLS \\
\hline
$l=40$ ms & $0.0$ & $0.0$ & $2.0 \times 10^{-4}$ & $6.6 \times 10^{-1}$  & $5.98$ & $5.98$ & $5.98$ & $5.98$ \\
$l=50$ ms & $0.0$ & $0.0$ & $2.8 \times 10^{-3}$ & $2.67$ & $9.22$ & $9.22$ & $9.80$ & $9.80$ \\
$l=100$ ms & $1.1 \times 10^{-2}$ & $1.1 \times 10^{-2}$ & $1.05$ & $9.91$ & $10.0$ & $10.0$ & $10.0$ & $10.0$ \\
\hline
\end{tabular}
\end{center}
\vspace{-15pt}
\end{table}

\subsection{Online Setting}

In this section, the proposed DDPG agent is evaluated.
\textcolor{black}{We assume that the latency constraint $l$ of each arrived task follows a uniform distribution in $[l_\text{low}, l_\text{high}]$ \cite{8493149}.
Two different task arrival processes are considered.
For Bernoulli-based task arrival, the probability of a task arriving in a typical time slot is $p_\text{arrive}$ if the time slot is before the latency constraint of the last arrived task, and zero otherwise.
We also consider a task arrival process that for each user, once the last arrived task reaches its latency constraint, a new task will arrive immediately at the next time slot (can be viewed as a special case of the Bernoulli-based arrival with $p_\text{arrive}=1$).
}
We use two 3-layer multilayer perceptron (MLP) models in our DDPG agent, one is the actor network and the other is the critic network. 
The two MLPs have the same architecture, where each hidden layer has 128 hidden nodes.
The parameters used in the experiment and DDPG training are listed in Table \ref{parameter_online}.

\begin{table}[!t]
\setlength{\abovecaptionskip}{2pt}
\setlength{\belowcaptionskip}{2pt}
\caption{System Parameters of the Online Setting and DDPG Training}
\label{parameter_online}
\begin{center}
\begin{tabular}{|c |c |c |c |c | c|}
\hline
\textbf{Parameter} & \textbf{Value} & \textbf{Parameter} & \textbf{Value} & \textbf{Parameter} & \textbf{Value} \\
\hline
\textcolor{black}{$T$} & \textcolor{black}{25 ms} & Episode length & 1000 s & $[l_\text{low}, l_\text{high}]$  &  mobilenet: $[0.05, 0.2]$ s; 3dssd: $[0.25, 1.0]$ s  \\
Optimizer  & Adam &  Target smoothing & 0.005 & \textcolor{black}{$p_{\text{arrive}}$} & \textcolor{black}{mobilenet: 0.25; 3dssd: 0.05}   \\
Batch size & 128 & Exploration noise & 0.1 & Learning rate & actor: 0.0001; critic: 0.001  \\
Discount & 0.99 & Updates per step & 200 & Reply buffer size & 1000000  \\
\hline
\end{tabular}
\end{center}
\vspace{-15pt}
\end{table}

\textcolor{black}{We denote the proposed policy as DDPG-OG, and compare it  with the following benchmarks:}
\begin{itemize}
    \item
    All users locally conduct the inference tasks, and is denoted by LC.
    \item
    \textcolor{black}{Calling IP-SSA or OG in some fixed values of time window (TW) is also considered.
    For example, $\text{TW}=2$ means that once the edge server completes the offloaded tasks and becomes idle, it will call IP-SSA or OG again after waiting for 2 time slots.}
    \item
    \textcolor{black}{A DDPG agent that is trained to call the IP-SSA algorithm (denoted by DDPG-IP-SSA).}
\end{itemize}

\begin{figure}[!t]
\setlength{\abovecaptionskip}{2pt}
\setlength{\belowcaptionskip}{2pt}
\centering
\subfloat[]
{\includegraphics[width=0.32\linewidth]{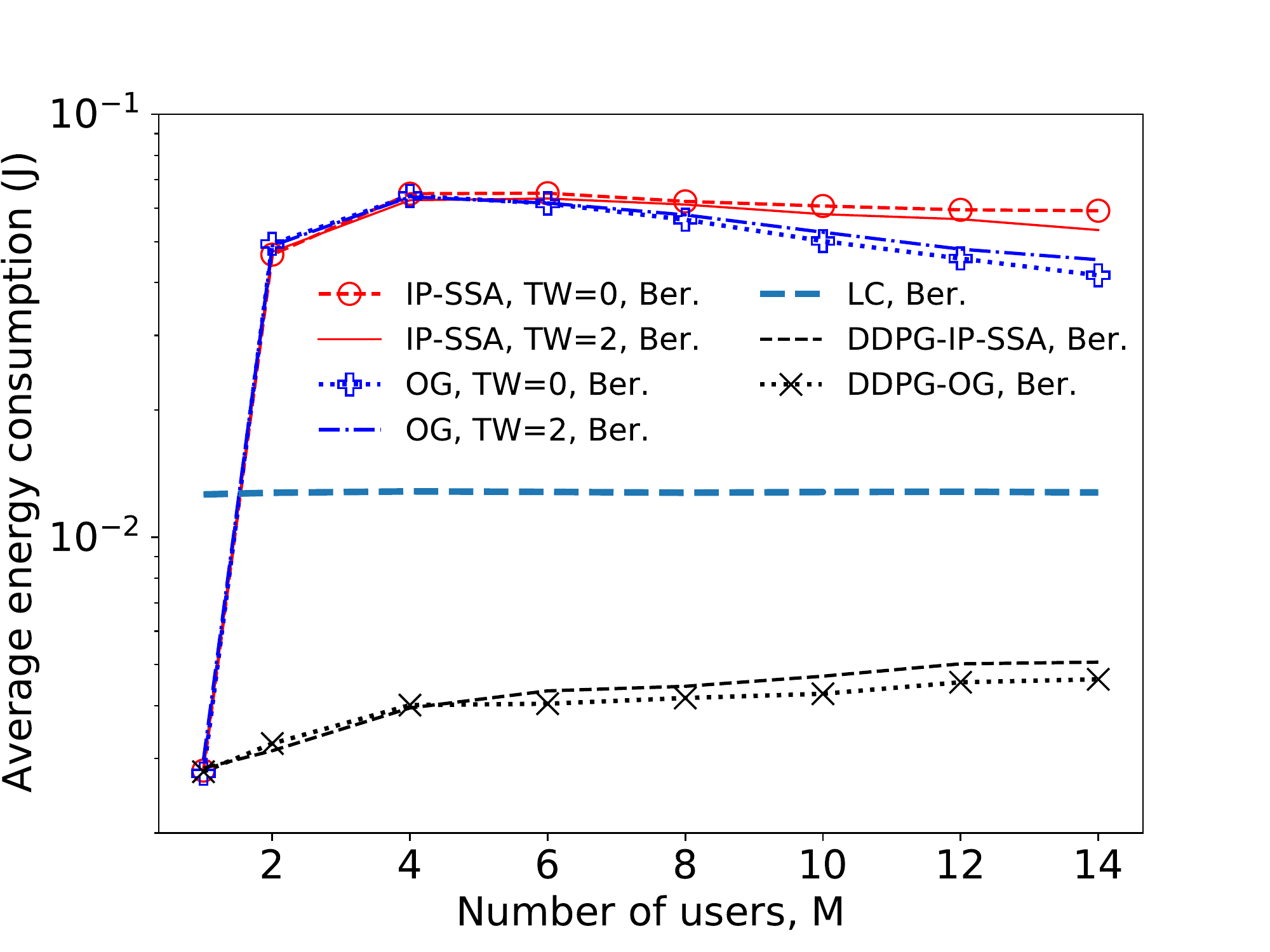}
\label{online_3dssd_Bernoulli}}
\hfil
\subfloat[]
{\includegraphics[width=0.32\linewidth]{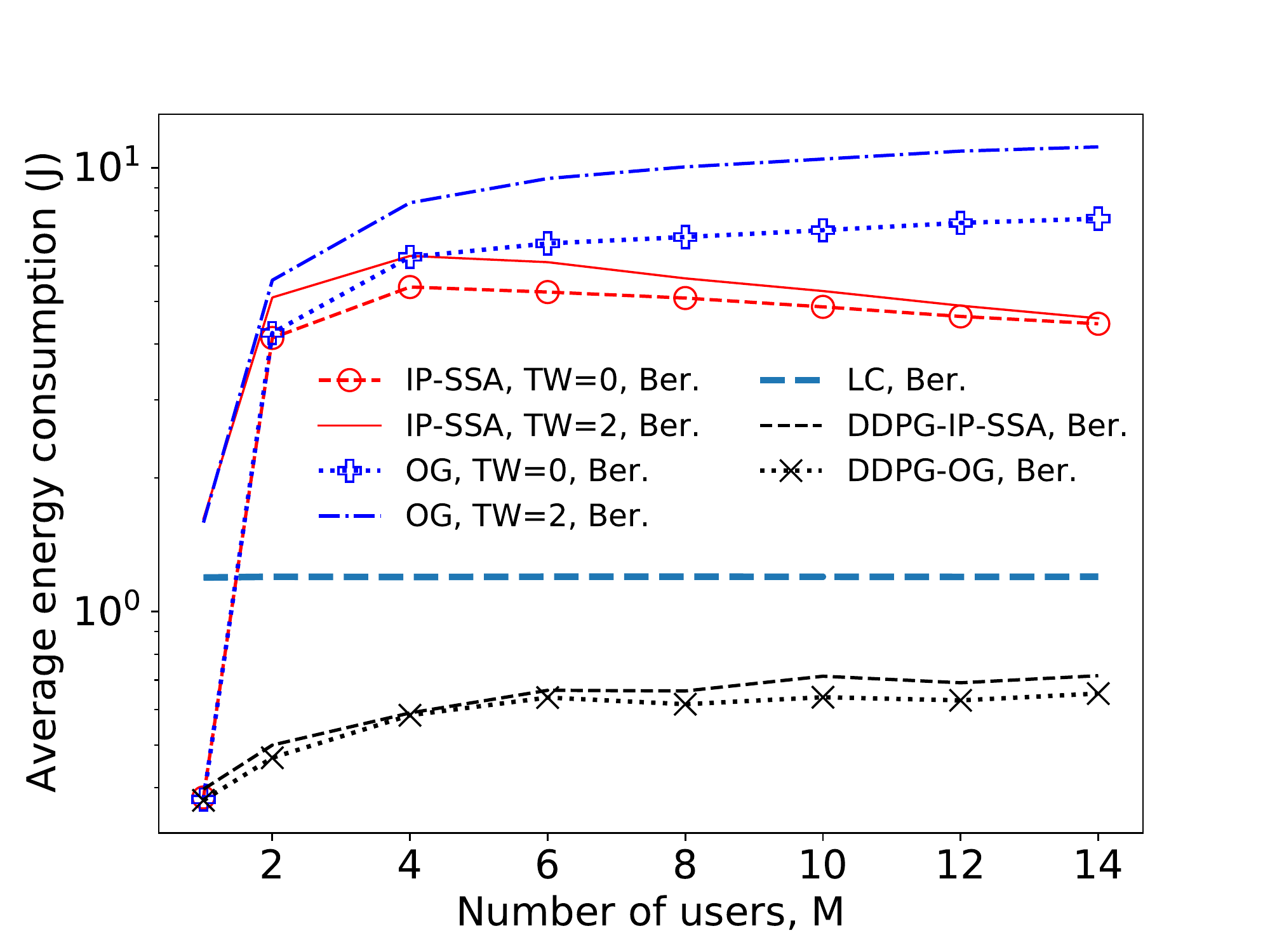}
\label{online_mobilenetv2_Bernolli}}
\hfil
\subfloat[]{
\includegraphics[width=0.32\linewidth]{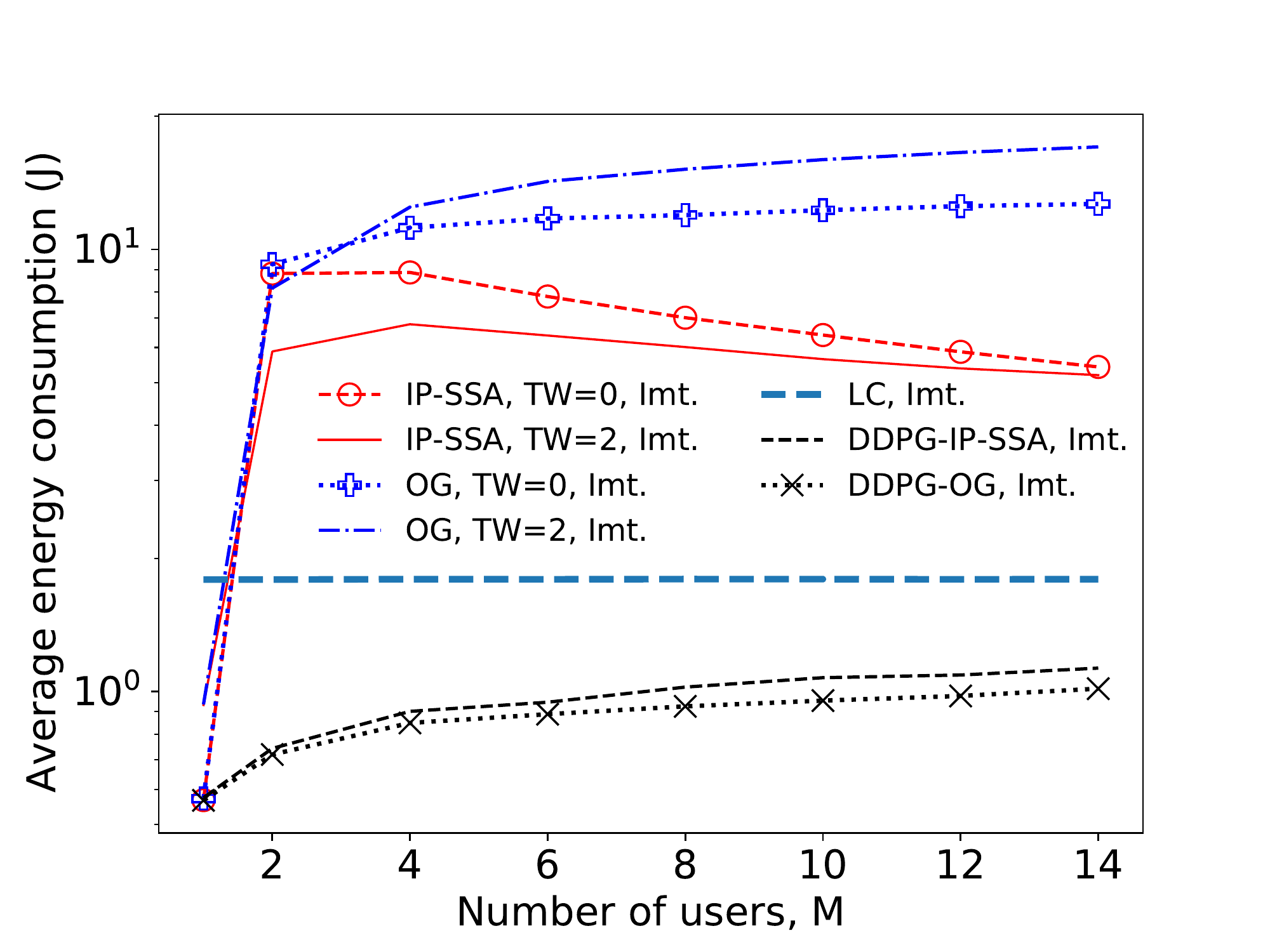}
\label{online_mobilenetv2_Intermediate}}
\hfil

\caption{\textcolor{black}{The average energy consumption per user per slot for all policies. Fig.\ref{online} (a) is for 3dssd, and Fig.\ref{online} (b), (c) are for mobilenet-v2. The task arrival process in Fig.\ref{online} (a), (b) is the Bernoulli-based arrival process (Ber. for short), and the task arrives immediately after the latency constraint of the last task in Fig.\ref{online} (c) (Imt. for short).}}
\label{online}
\vspace{-10pt}
\end{figure}

\begin{table}[!t]
\setlength{\abovecaptionskip}{2pt}
\setlength{\belowcaptionskip}{2pt}
\caption{\textcolor{black}{Average Results Under the Online Setting with $M=14$ and the Bernoulli-based Arrival}}
\label{results_M_14_Ber_arrival}
\begin{center}
\begin{tabular}{c c c c c c c}
\hline
& \multicolumn{3}{c}{3dssd} & \multicolumn{3}{c}{mobilenet-v2} \\
\cline{2-7}
& DDPG-OG & DDPG-IP-SSA & OG, TW=0 & DDPG-OG & DDPG-IP-SSA & OG, TW=0\\
\hline
Latency of DDPG (ms) & 0.24 & 0.24 & N.A. & 0.26 & 0.23 & N.A. \\
Latency of offline Alg. (ms) & 1.71  & 0.29  & 6.04  &  6.55  & 0.62  & 28.65 \\
Number of tasks  & 2.78 & 3.11 & 4.34 & 3.99 & 3.93 & 6.45 \\
Number of tasks per group & 2.56 & N.A. &  2.57 & 2.97 & N.A. & 2.08\\
\hline
\end{tabular}
\end{center}
\vspace{-15pt}
\end{table}

\textcolor{black}{
Fig. \ref{online} shows the average energy consumption per user per slot in one episode of all policies, after 500 episodes of training for each DDPG agent.
While Table \ref{results_M_14_Ber_arrival} shows the average execution latency of the DDPG agent and offline algorithms, the average number of tasks when calling the offline algorithms, and the average number of tasks in each group for the OG algorithm (if exists).
First, we notice that the fixed time window does not perform well when $M\geq 2$, since the fixed time window cannot seek the balance between serving arrived tasks and reserving resources for future tasks.
As shown in Table \ref{results_M_14_Ber_arrival}, 
for OG with $\text{TW}=0$, the average number of tasks when calling OG is much higher than that of DDPG-OG, indicating that the edge occupation period is too long.
However, via the proposed two-dimensional action, the proposed DDPG agent can adaptively balance the trade-offs between the waiting latency and the batch size, and between the processing time of current batch and the idle period.
As a result, under both the Bernoulli-based and immediate task arrivals, DDPG-based policies outperform other baselines.
Moreover, the OG algorithm can derive the optimal grouping policy, and the tasks with loose latency constraints can be processed in different batches with the tasks that have stringent latency constraints, which can improve the performance compared to IP-SSA.
The performance gap between DDPG-OG and DDPG-IP-SSA increases with the number of users.
When $M=14$, DDPG-OG can save up to 8.92\% and 8.85\% user energy consumption compared to DDPG-IP-SSA for 3dssd and mobilenet-v2 under the Bernoulli-based task arrival.
On the other hand, the main drawback of DDPG-OG is the high execution latency of OG.
As shown in Table \ref{results_M_14_Ber_arrival}, when the number of tasks when calling OG is large, the execution latency may exceed the time slot length (e.g., OG with $\text{TW}=0$ for mobilenet-v2).
Therefore, in the scenarios that the number of users is very large, DDPG-IP-SSA is preferred due to its low complexity.
}

\section{Conclusion}
In this paper, we have proposed a framework to jointly optimize DNN inference task offloading and offloaded task scheduling, for multi-user co-inference with batch processing capable edge server.
The problem of user energy consumption minimization under inference latency constraints is systematically solved, for both the offline and online scenarios.
Specifically, we propose IP-SSA that offloads sub-tasks of each user independently and schedules all the same sub-tasks in the same batch for tasks with the same latency constraint, and OG that groups the tasks with similar latency constraints for tasks with different latency constraints.
The experiment results reveal that IP-SSA and OG can greatly reduce user energy consumption by batch processing.
Further, DDPG-OG is proposed for the online scenario, where an RL agent is trained to control the trade-off between serving the arrived tasks and reserving resources for future tasks via the proposed two-dimensional control.
\textcolor{black}{As future work, large-scale edge inference systems with multiple servers can be further considered, in which the low-complexity distributed algorithms for user association, load balancing, and the queueing scheduling of batch processing might prove important.}

\appendices
\section{Proof of Theorem \ref{theorem1}} \label{appendix1}
Consider a typical optimal solution $x^*, s^*, t^*, f^*$, we prove that after the following modifications, it can be converted to an optimal solution that satisfies Theorem \ref{theorem1}.

First, consider a set of new batches, with starting time $s'_n, n\in \{1, \dots, N\}$, which completes the entire inference task just at the latency constraint $l$, i.e.,
\begin{equation}
\left\{
\begin{aligned}
    s'_N &=  l - F_N(1), \\
    s'_{N-1}   &=  s'_N - F_{N-1}(1), \\
    &\vdots \\
    s'_1  &= s'_2 - F_1(1).
\end{aligned}
\right.
\end{equation}
Since the edge server has stronger computing capability than mobile devices, $s'_n$ can be treated as the latest staring time that can ensure the latency constraint.

For those users that offload the $N$-th sub-task, i.e., $\mathcal{M}_N = \{m|x^*_{m,N,0}=0\}$, the starting time of the batches that process the $N$-th sub-task in the original solution is no later than $s'_N$, i.e.,
\begin{equation}
    s^*_k \leq s'_N, \ \forall k \in \{k|x^*_{m,N,k}=1, m\in\mathcal{M}_N \ \text{and} \ k\geq 1 \}.
\end{equation}
Therefore, we can let the users that complete the $(N-1)$-th sub-task earlier than $s'_N$ to wait, and aggregate all $N$-th sub-tasks into a batch with starting time $s'_N$.
Such modification ensures the latency constraint, and does not increase the user energy consumption.
As a result, Theorem \ref{theorem1} holds for all $N$-th sub-tasks after taking such modification.

Then, we complete the proof by induction.
Suppose Theorem \ref{theorem1} holds for all sub-tasks after the $n$-th sub-task.
Here, we still use $x^*, s^*, t^*, f^*$ to denote the optimal solution after the modification of the $n$-th to $N$-th sub-tasks.
For the $(n-1)$-th sub-task, one of the following holds: 
\begin{itemize}
    \item Both two sub-tasks are locally processed;
    \item The $(n-1)$-th sub-task is locally processed and the $n$-th is offloaded;
    \item Both two sub-tasks are offloaded;
    \item The $(n-1)$-th sub-task is offloaded and the $n$-th is locally processed.
\end{itemize}
It is obvious that for the first two cases, Theorem \ref{theorem1} holds for all sub-tasks after the $(n-1)$-th sub-task.
For the third case, since the starting time of the batches that process the $(n-1)$-th sub-task is no larger than $s'_{n-1}$, we can also aggregate all $(n-1)$-th sub-task into a batch with starting time $s'_{n-1}$, just like the previous modification to the $N$-th sub-task.
For the last case, assume that for the $m$-th user, the sub-tasks from $n$ to $n'$ are locally processed, and the sub-tasks from $n'+1$ to $N$ are offloaded.
Since Theorem \ref{theorem1} holds for all sub-tasks after the $n$-th sub task, we have
\begin{equation}
    t^*_{m,n-1} + \frac{B_{n-1}}{R^\text{d}_m} + \sum_{i=n}^{n'} \frac{A_i}{f^*_m} + \frac{B_{n'}}{R^\text{u}_m} \leq  s'_{n'+1},
\end{equation}
where the left-hand side is the ready time of the $(n'+1)$-th sub-task, and the right hand is the batch starting time of the $(n'+1)$-th sub-task.
Therefore, we have 
\begin{equation}
    \begin{split}
    \label{t_m,n-1}
    t^*_{m,n-1} & \leq  s'_{n'+1} - \left(\frac{B_{n-1}}{R^\text{d}_m} + \sum_{i=n}^{n'} \frac{A_i}{f^*_m} + \frac{B_{n'}}{R^\text{u}_m} \right)  \\
    & < s'_{n'+1} - \sum_{i=n}^{n'} \frac{A_i}{f^*_m}  \\
    & \leq s'_{n'+1} - \sum_{i=n}^{n'} F_i(1)  \\
    & = s'_n,
    \end{split}
\end{equation}
where the last inequality is because the edge server has stronger computing capabilities than the mobile devices.
We can offload all sub-tasks between $n$ and $n'$ of the $m$-th user, and the modified solution is still feasible due to \eqref{t_m,n-1}.
Since the user energy consumption of the modified solution is strictly reduced due to the offloading, it conflicts with the assumption that $x^*, s^*, t^*, f^*$ is optimal.
Therefore, the forth case will not happen.
As a result, Theorem \ref{theorem1} holds for all sub-tasks after the $(n-1)$-th sub task after the modification for the first three cases.

Finally, Theorem \ref{theorem1}. (1) and (2) are proved by induction.
Furthermore, due to that the user energy consumption is a decreasing function of local computing frequency $f$, Theorem \ref{theorem1}. (3) is proved immediately.

\section{Proof of Theorem \ref{subsequent_users}} \label{appendix2}
Suppose for an optimal solution $\mathcal{G}^*_1, \mathcal{G}^*_2, \dots, \mathcal{G}^*_g$, there exist $j\in \mathcal{G}^*_p, k\in \mathcal{G}^*_q$, such that $j>k$ and $p<q$.
Then we can move user $j$ to group $\mathcal{G}_q$ from group $\mathcal{G}_p$.
We first prove that the new solution ${\mathcal{G}^*_1}', {\mathcal{G}^*_2}', \dots, {\mathcal{G}^*_g}'$ after the movement is feasible.
Due to \eqref{tildel} and \eqref{nonoverlap}, we have $\tilde{l}_p < \tilde{l}_q \leq l_k \leq l_j$. 
Therefore, moving user $j$ to group $\mathcal{G}_q$ does not change the latency constraints of both two groups, i.e., ${\tilde{l}_p}'=\tilde{l}_p$ and ${\tilde{l}_q}'=\tilde{l}_q$, and thus the feasibility holds.
On the other hand, since ${\tilde{l}_p}'< {\tilde{l}_q}'$, the movement does not increase the user energy consumption, and thus the new solution after the movement is still optimal.

We can repeat such movement several times, until there is no $j\in \mathcal{G}^*_p$ and $k\in \mathcal{G}^*_q$ satisfy that $j>k$ and $p<q$.
Obviously, the number of repetitions is finite.
Therefore, we can prove that there exists an optimal solution, for any two users in different groups, the index of the user in the former group is smaller than that of the user in the latter group.
This result is equivalent to Theorem \ref{subsequent_users}.

\bibliographystyle{IEEEtran}
\bibliography{reference}

\end{document}